\documentclass[11pt]{article}

\newif\ifnotes
\notestrue

\title{Efficient Interactive Coding Achieving Optimal Error Resilience Over the Binary Channel}
\author{Meghal Gupta\thanks{Email: \texttt{meghal@mit.edu}}\\Microsoft Research 
\and Rachel Yun Zhang\thanks{Email: \texttt{rachelyz@mit.edu}}\\Massachusetts Institute of Technology}
\date{\today}

\usepackage{hyperref}
\usepackage{fullpage}
\usepackage{amssymb}
\usepackage{amsmath}
\usepackage{amsthm}
\usepackage{amsfonts}
\usepackage{bm}
\usepackage{enumitem}
\usepackage{color}
\usepackage{comment}
\usepackage[capitalize]{cleveref}
\usepackage[dvipsnames]{xcolor}
\usepackage{float}
\usepackage[T1]{fontenc}
\usepackage{todonotes}
\usepackage{asymptote}
\usepackage{mdframed}
\usepackage[most]{tcolorbox}
\usepackage{hyperref}
\usepackage{enumitem}
\usepackage{framed}
\usepackage{mdframed}
\usepackage{scrextend}
\usepackage{multirow}
\usepackage{ifthen}
\usepackage{bbm}
\usepackage{frcursive}

\usepackage[T1]{fontenc}

\definecolor{denim}{rgb}{0.08, 0.38, 0.74}
\definecolor{periwinkle}{rgb}{0.6, 0.6, 0.95}
\definecolor{wildblueyonder}{rgb}{0.64, 0.68, 0.82}
\definecolor{wisteria}{rgb}{0.91, 0.72, 1.00}
\definecolor{thistle}{rgb}{0.85, 0.75, 0.85}
\definecolor{byzantium}{rgb}{0.44, 0.16, 0.39}
\definecolor{deeplilac}{rgb}{0.6, 0.33, 0.73}
\definecolor{jazzberryjam}{rgb}{0.55, 0.04, 0.37}

\definecolor{fireenginered}{rgb}{0.81, 0.09, 0.13}
\definecolor{deepcarrotorange}{rgb}{0.91, 0.41, 0.17}
\definecolor{mangotango}{rgb}{1.0, 0.51, 0.26}

\usepackage{hyperref}
\hypersetup{
    colorlinks=true,
    linkcolor=fireenginered,
    filecolor=fireenginered,
    citecolor=deepcarrotorange,
    urlcolor=deepcarrotorange,
}
\usepackage[hyperpageref]{backref}

\newtheorem{theorem}{Theorem}[section]
\newtheorem{lemma}[theorem]{Lemma}

\newtheorem{corollary}[theorem]{Corollary}

\theoremstyle{definition}
\newtheorem{definition}[theorem]{Definition}
\newtheorem{remark}[theorem]{Remark}

\Crefname{theorem}{Theorem}{Theorems}
\Crefname{claim}{Claim}{Claims}
\Crefname{lemma}{Lemma}{Lemmas}
\Crefname{proposition}{Proposition}{Propositions}
\Crefname{corollary}{Corollary}{Corollaries}
\Crefname{definition}{Definition}{Definitions}

\newcommand{\ECC}{\mathsf{ECC}}

\newcommand{\op}{\mathsf{op}}

\newcommand{\rewind}{\,\leftarrow}

\newcommand{\asked}{\mathsf{asked}}

\newcommand{\true}{\mathsf{true}}
\newcommand{\false}{\mathsf{false}}

\newcommand{\PikC}{\mathsf{C}}

\newcommand{\poly}{\text{poly}}
\newcommand{\polylog}{\text{polylog}}

\newcommand{\PT}{\mathcal{PT}}
\newcommand{\PikCDec}{\mathsf{CDec}}

\newcommand{\val}{\mathsf{val}}

\newcommand{\bbN}{\mathbb{N}}

\newcommand{\bbR}{\mathbb{R}}

\newcommand{\bbE}{\mathbb{E}}

\newcommand{\cE}{\mathcal{E}}

\newcommand{\cL}{\mathcal{L}}

\newcommand{\cP}{\mathcal{P}}

\newcommand{\cS}{\mathcal{S}}
\newcommand{\cT}{\mathcal{T}}

\makeatletter
\newcommand{\customlabel}[2]{%
   \protected@write \@auxout {}{\string \newlabel {#1}{{#2}{\thepage}{#2}{#1}{}} }%
   \hypertarget{#1}{#2}
}
\makeatother

\newcommand{\protocol}[3]{
    \stepcounter{figure}
    \vspace{0.15cm}
    { \small
    \begin{tcolorbox}[breakable, enhanced, colback=mangotango!10]
    \begin{center}
    {\bf \underline{Protocol~\customlabel{prot:#2}{\thefigure}: #1}}
    \end{center}
    
    #3
    \end{tcolorbox}
    }
}

\newcommand\numberthis{\addtocounter{equation}{1}\tag{\theequation}}

\newcounter{casenum}
\newenvironment{caseof}{\setcounter{casenum}{0}}{\vskip.5\baselineskip}
\newcommand{\case}[2]{
    \refstepcounter{casenum}
    \ifthenelse{\equal{\value{casenum}}{0}}{
    \vskip.5\baselineskip\par\noindent
    }{}
    {\it Case \arabic{casenum}:} {\it #1}
    \vskip0.1\baselineskip
    \begin{addmargin}[1.5em]{1em}
    #2
    \end{addmargin}
}

\newcounter{subcasenum}
\newenvironment{subcaseof}{\setcounter{subcasenum}{0}}{\vskip.5\baselineskip}
\newcommand{\subcase}[2]{
    \refstepcounter{subcasenum}
    \vskip.5\baselineskip\par\noindent 
    {\it Subcase \arabic{casenum}.\arabic{subcasenum}:} {\it #1} \vskip0.1\baselineskip
    \begin{addmargin}[1.5em]{1em}
    #2
    \end{addmargin}
}
  
\newcounter{casenumb}
\newenvironment{caseofb}{\setcounter{casenumb}{0}}{\vskip.5\baselineskip}
\newcommand{\caseb}[2]{
    \refstepcounter{casenumb}
    \vskip.5\baselineskip\par\noindent 
    {\bf Case \arabic{casenumb}:} {\it #1} \vskip0.1\baselineskip
    \begin{addmargin}[1.5em]{1em}
    #2
    \end{addmargin}
}

\newcounter{subcasenumb}
\newenvironment{subcaseofb}{\setcounter{subcasenumb}{0}}{\vskip.5\baselineskip}
\newcommand{\subcaseb}[2]{
    \refstepcounter{subcasenumb}
    \vskip.5\baselineskip\par\noindent 
    {\bf Subcase \arabic{casenumb}.\arabic{subcasenumb}:} {\it #1} \vskip0.1\baselineskip
    \begin{addmargin}[1.5em]{1em}
    #2
    \end{addmargin}
}

\begin{document}

\sloppy
\maketitle
\begin{abstract}
Given a noiseless protocol $\pi_0$ computing a function $f(x, y)$ of Alice and Bob's private inputs $x, y$, the goal of interactive coding is to construct an \emph{error-resilient} protocol $\pi$ computing $f$ such that even if some fraction of the communication is adversarially corrupted, both parties still learn $f(x, y)$. Ideally, the resulting scheme $\pi$ should be positive rate, computationally efficient, and achieve optimal error resilience.

While interactive coding over large alphabets is well understood, the situation over the binary alphabet has remained evasive. At the present moment, the known schemes over the binary alphabet that achieve a higher error resilience than a trivial adaptation of large alphabet schemes are either still suboptimally error resilient~\cite{EfremenkoKS20b}, or optimally error resilient with exponential communication complexity~\cite{GuptaZ22a}. In this work, we construct a scheme achieving optimality in all three parameters: our protocol is positive rate, computationally efficient, and resilient to the optimal $\frac16 - \epsilon$ adversarial errors.

Our protocol employs a new type of code that we call a \emph{layered code}, which may be of independent interest. Like a tree code, a layered code allows the coder to encode a message in an online fashion, but is defined on a graph instead of a tree. 

\end{abstract}
\thispagestyle{empty}
\newpage

\tableofcontents
\pagenumbering{roman}
\newpage
\pagenumbering{arabic}

\section{Introduction}
Interactive coding is an interactive analogue of error correcting codes \cite{Shannon48, Hamming50} that was introduced in the  seminal work of Schulman~\cite{Schulman92,Schulman93,Schulman96} and has been an active area of study since. While error correcting codes address the problem of sending a {\em message} in a way that is resilient to error, interactive coding addresses the problem of converting an {\em interactive protocol} to an error resilient one.

Suppose two parties, Alice and Bob, each with a private input, engage in a protocol $\pi_0$ to jointly compute a function $f$ of their private inputs. Given such a protocol $\pi_0$, can we design a protocol computing $f$ that is:
\begin{enumerate}[label={(\roman*)}]
    \item\label{itm:pos-rate} positive rate, i.e. $|\pi| = O(|\pi_0|)$ where $|\pi|, |\pi_0|$ denote the communication complexity of $\pi, \pi_0$,
    \item\label{itm:eff-dec} computationally efficient,
    \item\label{itm:opt-error} resilient to the maximal possible fraction of adversarial errors?
\end{enumerate}
The protocol should have a fixed number of rounds and speaking order. This parallels the notion of an efficiently encodable/decodable error correcting code with maximal distance.

The first positive rate interactive coding scheme, presented by Schulman~\cite{Schulman96}, was resilient to $\frac1{240}$\footnote{Whenever we say that a protocol has resilience $r \in [0, 1]$ in the introduction and overview, we mean that for any $\epsilon$, there exists an instantiation that achieves resilience $r - \epsilon$.} adversarial errors (bit flips) over the binary channel but is exponentially inefficient, thus satisfying \ref{itm:pos-rate} but not \ref{itm:eff-dec} or \ref{itm:opt-error}. Many works since then sought to improve upon this scheme in computational efficiency and/or error resilience.



When the encoding alphabet is \emph{large} constant sized, Braverman and Rao~\cite{BravermanR11} first studied the problem of optimal error resilience. They constructed a large alphabet protocol achieving $\frac14$ error resilience, which they also showed to be optimal. Unfortunately, their protocol did not achieve computational efficiency~\ref{itm:eff-dec}. Computationally efficient schemes were not known until the work of~\cite{BrakerskiK12}, who converted the $\frac14$-error resilient, inefficient protocol to an efficient one achieving only $\frac1{16}$ error resilience. Finally, the work of~\cite{GhaffariH13} attained the best of both worlds: they constructed a protocol that was simultaneously efficiently decodable and resilient to $\frac14$ error, thus satisfying all three criteria.


On the other hand, over the binary alphabet, optimal interactive coding has remained less well understood. By simply replacing every letter of a large alphabet with its binary encoding, the large alphabet protocols give rise to efficient, positive rate interactive coding schemes achieving an error resilience of $\frac18$. By contrast, the best known upper bound on error resilience is $\frac16$~\cite{EfremenkoGH16}. There are two works improving the error resilience beyond $\frac18$. The first is~\cite{EfremenkoKS20b}. Their protocol is resilient to $\frac5{39}$ error, and is positive rate but inefficient. The second is~\cite{GuptaZ22a}, which constructs a scheme achieving the optimal $\frac16$-error resilience. However, both the communication and computational complexity can be up to exponential in the length of $\pi_0$. It thus remained open whether there exists a scheme resilient to the maximal amount of error, while also being positive rate and efficient.

In this work, we construct precisely such a scheme. Our result, along with comparison to existing work, is given in Figure~\ref{fig:binary-schemes}.

\begin{theorem}
    For any $\epsilon > 0$ and any interactive binary protocol $\pi_0$ computing a function $f(x,y)$ of Alice and Bob's private inputs $x, y$, there exists a non-adaptive interactive binary protocol $\pi$ computing $f(x,y)$ that is resilient to $\frac16 - \epsilon$ adversarial erasures. The communication complexity is $O_\epsilon(|\pi_0|)$ and the computational complexity is $\tilde{O}_\epsilon(|\pi_0|)$.
\end{theorem}

\begin{figure}[h!]
\centering
\begin{tabular}{|c|c|c|c|}
    \hline
    Protocol & Positive Rate? & Efficient? & Error Resilience \\
    \hline
    \cite{GhaffariH13} & yes & yes & $1/8$ \\
    \cite{EfremenkoKS20b} & yes & no & $5/39$ \\
    \cite{GuptaZ22a} & no & no & $1/6$ (optimal) \\
    This work & yes & yes & $1/6$ (optimal) \\
    \hline
\end{tabular}
\caption{Interactive coding schemes over the binary channel}
\label{fig:binary-schemes}
\end{figure}


\paragraph{Layered Codes.}
Our protocol crucially relies on a new type of code that we call a \emph{layered code}, which generalizes a tree code. Recall that tree codes~\cite{Schulman93,Schulman96} are error correcting codes that can be updated in an online manner: the $i$'th symbol in a codeword is dependent only on the first $i$ characters in the message. One can view a tree code as an assignment of code symbols $\Sigma_{code}$ to the edges of the infinite $|\Sigma_{mes}|$-ary rooted tree, where $\Sigma_{mes}$ is the alphabet of the message text. To encode a message $\in \Sigma_{mes}^*$, one simply follows the rooted path specified by the message and reads the code symbols off the edges. 

Instead of being defined on trees, layered codes are an assignment of $\Sigma_{code}$ to a certain kind of graph called \emph{layered graphs}. A layered graph is a directed graph where vertices are partitioned into layers such that there is only one vertex (the root node) in layer $0$, and each vertex in layer $i$ has out-edges labeled with $\Sigma_{mes}$ to vertices in layer $i+1$.\footnote{Note that tree codes are layered codes, so our notion of a layered code generalizes tree codes.} As with tree codes, to encode a message $\in \Sigma^*_{mes}$, one simply follows the rooted path specified by the message and reads the code symbols off the edges.

In the literature, tree codes with a variety of distance or decoding properties have been studied \cite{Schulman96, GellesMS11, BravermanE14}. In our protocol, however, we will need our layered codes to satisfy a certain new special property we call \emph{sensitivity}. Intuitively, sensitivity means that a corrupted layered code can be \emph{entirely} decoded correctly as long as the latest symbol was received correctly. More precisely, we show that:





\begin{theorem} [Informal]
    There exists a layered code (i.e. an assignment of labels to a layered graph) with the following property: for any string $w \in \Sigma_{code}^n$ and message text $x \in \Sigma_{mes}^n$, $w[1:i]$ uniquely decodes to $v(x[1:i])$ for almost every $i$ for which $w[i] = \PikC(x)[i]$. Here, $v(x[1:i])$ denotes the vertex at the end of the rooted path specified by $x[1:i]$.
\end{theorem}

Layered codes may be of independent interest, beyond the application to our protocol. One might also want to generalize more of the study of tree codes to the graph setting. We leave this as an open topic, and discuss this further in Section~\ref{sec:gcode-discussion}.

\subsection{Related Work}

Our work relates primarily to the fields of interactive coding and tree codes.
Besides the works we have already discussed, we mention the following related works.

\subsubsection{Interactive Coding}
Non-adaptive interactive coding (when the protocol is fixed length and fixed speaking order) was studied starting with the seminal works of Schulman~\cite{Schulman92,Schulman93,Schulman96} and continuing in a prolific sequence of followup works, including~\cite{BravermanR11,Braverman12,BrakerskiK12,BrakerskiN13,Haeupler14,BravermanE14,DaniHMSY15,GellesHKRW16, GellesH17,EfremenkoGH16,GhaffariH13,GellesI18,EfremenkoKS20b,GuptaZ22a}.

We note that there are many other works studying variations upon this original interactive coding setup, including adaptive and multi-party schemes. We refer the reader to an excellent survey by Gelles~\cite{Gelles-survey} for an extensive list of related work. 

\paragraph{Other binary schemes resilient to $\frac16$ error.} 
\cite{EfremenkoGH16} studies interactive coding over the \emph{feedback} channel. Over the feedback channel, Alice and Bob are given the extra power to know, instantly, what the other party received at the other end of the channel when they send a message. In this setting, \cite{EfremenkoGH16} constructs a positive rate, efficient protocol resilient to $\frac16$ error, which is optimal in the feedback setting as well. By contrast, we achieve $\frac16$-error resilience with positive rate in the standard setting \emph{without} feedback.

The protocol of \cite{EfremenkoGH16} relies on feedback for a ``guess'' of the transcript so far, and then the party responds according to whether or not they agree with this guess. The protocol of \cite{GuptaZ22a} (achieving $\frac16$ error resilience in channels without feedback, but inefficiently) also uses this idea, however providing (unreliable) feedback through future messages instead. One step in our protocol uses this idea as well, following the blueprint of the construction in \cite{GuptaZ22a}.

\paragraph{Efficiency.}
We also mention the work on obtaining interactive protocols that are \emph{efficient}: protocols where Alice and Bob can compute their next message and output their final answer in polynomial time. While Braverman and Rao's protocol~\cite{BravermanR11} is resilient to $\frac14$ corruption over a large alphabet and incurs only a constant blowup in communication complexity, the parties' computational efficiency incurs exponential blowup. 

The work of \cite{GhaffariH13} which draws inspiration from \cite{BrakerskiK12} addresses this problem. They provide an algorithm which takes a protocol and ``boosts'' it, lowering the computational complexity while increasing the alphabet size. We use a similar method to make our protocol computationally efficient while avoiding the alphabet blowup.

\subsubsection{Tree codes.}

Tree codes were first introduced by Schulman~\cite{Schulman93,Schulman96} and have been studied since in a variety of works~\cite{GellesMS11, Braverman12, MooreS14, FranklinGOS15, BravermanGMO15, Pudlak16, CohenHS18, BenyaacovCY21}.
Tree codes are a key ingredient in achieving constant rate interactive coding schemes. They also have important uses as streaming codes for both Hamming errors \cite{FranklinGOS15} and synchronization errors \cite{BravermanGMO15,HaeuplerS21}. Recently, there has been work towards finding explicit tree codes with a constant sized alphabet that are efficiently decodable and encodable \cite{CohenHS18, BenyaacovCY21}.

We specifically mention the concept of list tree codes introduced in \cite{BravermanE14}, which are the list-decoding analogue of error correcting codes in the tree code setting. Our concept of sensitive layered codes generalize and strengthen Braverman and Efremenko's definition of list tree codes.




\section{Technical Overview}

We begin by recalling at a high level the binary protocol of~\cite{GuptaZ22a}, which achieves optimal error resilience $\frac16 - \epsilon$, but whose communication complexity is quadratic in the input lengths.

Suppose Alice and Bob have private inputs $x, y \in \{ 0, 1 \}^n$. Consider the task of \emph{message exchange}, where the goal is for Bob to learn $x$ and for Alice to learn $y$. The protocol of~\cite{GuptaZ22a} is a $(\frac16-\epsilon)$-error resilient protocol achieving message exchange, where the communication complexity is $O_\epsilon(n^2)$.

The protocol works as follows. Alice and Bob each keep a track of a guess $\hat{y}$ or $\hat{x}$ for the other party's input, initially set to $\emptyset$, and a weight $w_A$ or $w_B$ indicating their confidence for their guess $\hat{y}$ or $\hat{x}$ respectively, initially set to $0$. 

The idea is that Alice can ask a \emph{question} by sending Bob her guess $\hat{y}$ encoded in an error correcting code. Bob can then send her an \emph{answer} telling her how to update $\hat{y}$ to bring it closer to his actual input $y$: append $0$ ($0$), append 1 ($1$), delete the last bit ($\rewind$), or ``bingo -- you got it right!'' ($*$). (This last instruction $*$ tells Alice to increase $w_A$. If Alice receives an instruction to modify $\hat{y}$ while $w_A > 0$, she decreases $w_A$ by $1$ instead.) Since Bob's answer is always one of four options, his possible answers can be made to be relative distance $\frac23$ apart (e.g. $000, 011, 101, 110$), so that the adversary would have to corrupt $\ge \frac13$ of Bob's bits sent (or $\frac16$ overall) to prevent Alice from making good updates to $\hat{y}$ (i.e. updates that get $\hat{y}$ closer to $y$).

Now, since both Alice and Bob have to learn the other's input, Alice and Bob \emph{simultaneously} ask a question and answer the other party's last question. In other words, Alice's message is always of the form $\ECC(\hat{y}, x^*, \delta)$, where $x^*$ is the question she just heard from Bob and $\delta$ is the instruction on how to update $x^*$ to bring it closer to $x$. Similarly, Bob's message is always of the form $\ECC(\hat{x}, y^*, \delta)$. Here, $\ECC$ is a code with certain distance properties, including that for any $x',y'$ the four codewords $\{ \ECC(x', y', 0), \ECC(x', y', 1), \ECC(x', y', \rewind), \ECC(x', y', *) \}$ should be pairwise relative distance $\frac23$ from each other. 

However, there are two problems with this current algorithm:
\begin{enumerate}[label={(\alph*)}]
    \item \label{item:problem1} The adversary can simultaneously corrupt both the question and answer in Bob's message $\ECC(\hat{x}, \hat{y}, \delta)$ by only corrupting $\frac12$ of the message, so that Alice receives an incorrect answer and thus makes a bad update for only $\frac12$ cost.
    \item \label{item:problem2} The adversary can partially corrupt Bob's message (so that the message Alice receives is not any codeword), so Alice does not know what question to answer.
\end{enumerate}

The algorithm of \cite{GuptaZ22a} fixes these problems with two additional rules.

\begin{itemize}
    \item When Alice receives a message $\ECC(x', \hat{y}, \delta')$, she usually only updates with probability $0.5$. However, if $x'=x$ (i.e. Bob has already figured out her input), she updates with probability $1$.
    \item When Alice receives a partially corrupted message where she cannot determine what question to answer, she defaults to sending $\ECC(\hat{y}, x, *)$. Correspondingly, when Bob receives any message $\ECC(y', x', *)$ where the update instruction is $*$, he updates $\hat{x}$ to be closer to $x'$. 
\end{itemize}

Both these new rules require one important fact: that Alice knows what Bob's correct output ought to be (her input $x$).
For us, we will be simulating a noiseless protocol $\pi_0$ where the final transcript depends on both parties' private inputs, so that neither Alice nor Bob knows what the correct final transcript ought to be. This is the main barrier to making the protocol of~\cite{GuptaZ22a} run in time $O_\epsilon(|\pi_0|^2)$ as opposed to in time $O_\epsilon(n^2)$. 

\subsection{Obtaining Communication Complexity $O_\epsilon(|\pi_0|^2)$} \label{sec:overview-pi^2}

The first modification we will make is to create an interactive coding scheme that can simulate general protocols, instead of just message exchange, in quadratic time. By doing this, we will obtain a protocol with communication complexity $O_\epsilon(|\pi_0|^2)$ instead of $O_\epsilon(n^2)$.

At a high level, in our protocol, in each message Alice and Bob either asks a question \emph{or} answers a received question, \emph{but not both}.
This is as opposed to the protocol of~\cite{GuptaZ22a}, in which question asking and answering are always done simultaneously. We remark that this removes issue~\ref{item:problem1} with the~\cite{GuptaZ22a} protocol, since now answers no longer have a question component so that all possible answers $\{ \ECC(r^*, 0), \ECC(r^*, 1), \ECC(r^*, \rewind), \ECC(r^*, \bullet) \}$ to the same question $r^*$ are distance $\frac23$ apart.

More concretely, Alice and Bob each keep track of a guess for the complete noiseless transcript, denoted $T_A$ or $T_B$ respectively, along with a weight $w_A$ or $w_B$ signaling how confident they are that the current transcript guess is correct. We have that $w = 0$ unless the corresponding transcript guess $T$ is complete, meaning $|T| = |\pi_0|$. Alice's transcript guess $T_A$ always has odd length, i.e. she is the last to speak, unless $T_A$ is a complete transcript or is the empty transcript. Similarly, Bob's transcript guess $T_B$ always has even length. Let $\cT$ denote the noiseless transcript, so that the goal is for Alice and Bob to have $T_A = T_B = \cT$ by the end of the protocol. In what follows, we describe the protocol from Alice's point of view, but Bob's behavior is equivalent.

Every round, Alice sends a message of the form $\ECC(T, \delta \in \{ 0, 1, \rewind, ? \})$, where $\delta = ?$ signals that she is asking a question and $\delta \in \{ 0, 1, \rewind \}$ signals that she is answering a question. Specifically, when Alice asks a question, she sends $\ECC(T_A, ?)$. She answers a question $T^*_B$ by sending $\ECC(T^*_B, \delta)$, where $\delta \in \{ 0, 1, \rewind \}$ is 
\begin{itemize}
    \item $\rewind$ if $T^*_B$ is not consistent with her own behavior on input $x$.
    \item her next message $0$ or $1$ given the consistent transcript prefix $T^*_B$ (if $T^*_B$ is a complete transcript, then her next message is just $1$).
\end{itemize} 
Here, $\ECC$ is a code satisfying that for any $T^*$ the four words $\ECC(T^*, 0)$, $\ECC(T^*, 1)$, $\ECC(T^*, \rewind)$, $\ECC(T^*, ?)$ have relative distance $\frac23$ and all other pairs of codewords are relative distance $\frac12$ apart. Such a code was shown to exist in~\cite{GuptaZ22a}.

Alice determines whether to ask or answer based on the message she just received:
\begin{itemize}
    \item As long as she receives an answer (not necessarily to the question she previously asked), she asks a question.
    \item Whenever Alice receives a question, she answers it. There is an exception, which is when the question received is a complete transcript consistent with Alice's own input $x$. In this case, Alice asks her own question. This mechanism allows Alice and Bob to switch who is asking vs. answering once the asking party has made sufficient progress and now knows $\cT$.
\end{itemize}

Furthermore, every time Alice receives a message from Bob, she needs to update $(T_A, w_A)$ accordingly:
\begin{itemize}
    \item When she receives an answer to her question $\ECC(T_A, \delta\in \{0,1\})$, she concatenates $\delta$ and her resulting next message to the end of $T_A$. (If $T_A$ is a complete transcript, she instead increments $w_A$.) 
    \item If she receives $\ECC(T_A, \rewind)$, assuming $w_A = 0$ she deletes the last two messages (one of hers and one of Bob's) from $T_A$, and otherwise if $w_A > 0$ she simply decreases $w_A$ by $1$. 
    \item If she receives a question $\ECC(T^*_B,?)$ from Bob, where $T^*_B$ corresponds to a complete transcript that is consistent with her input $x$, she updates $T_A$ to be one step closer to $T^*_B$ with $0.5$ probability. 
    
    There is an exception to this rule, which is when $T^*_B = T_A$. This can only happen if $T^*_B = T_A$ is either $\emptyset$ or a complete transcript, as in general $T_A$ is of odd length and $T_B$ is of even. In this case, with probability $1$ instead of $0.5$, Alice increases her weight $w_A$ on the transcript $T_A$ by $1$. This is because when $T_A = T_B = \cT$, we want both Alice and Bob to make more progress simultaneously.\footnote{The potential function we care about is $[\text{Alice's progress}] + \min\{ [\text{Bob's progress}], |\pi_0| \}$, so once Bob's progress is $\ge |\pi_0|$ signaling that $T_B = \cT$, we need Alice to be updating with probability $1$ each time she correctly receives Bob's message.} Similarly, Bob also needs to be updating with probability $1$ whenever he receives a question from Alice equal to $T_B$. 
    
    \item Otherwise, she does not update $T_A$ or $w_A$.
\end{itemize}

So far, we have described the protocol when the parties receive full codewords. When messages are \emph{partially corrupted} so that the received message is not a codeword, a party will default to asking a question with probability proportional to the distance from the nearest codeword, and otherwise employ the above behavior. This addresses issue~\ref{item:problem2}. We remark that the default message being a question is the second idea that allows us to escape from needing for Alice and Bob to know what the other party's output ought to be, since instead of defaulting to sending the answer $(x, *)$ or $(y, *)$ one now defaults to asking a question.



\subsection{Reducing the Communication Complexity to $O_\epsilon(|\pi_0|)$} \label{sec:overview-graph-codes}

Now that we have an optimally error resilient interactive coding scheme that can simulate protocols with $O_\epsilon(|\pi_0|^2)$ communication complexity, the next step is to reduce the communication complexity to $O_\epsilon(|\pi_0|)$. 

Currently, the quadratic factor in the communication complexity arises because we need $O_\epsilon(|\pi_0|)$ rounds to simulate the protocol, and in each round the parties are sending either their transcript guess or the transcript guess they are answering, both of which takes $O_\epsilon(|\pi_0|)$ bits. If we could reduce the amount of communication needed to send a transcript guess to $O_\epsilon(1)$, then we could achieve our desired $O_\epsilon(|\pi_0|)$ total communication.

Consider first the task of a party sending their own transcript guess as a question such that each message is only $O_\epsilon(1)$ bits. The traditional solution for this problem in interactive coding is to use \emph{tree codes}~\cite{Schulman93,Schulman96}, which are essentially error correcting codes that one can update in an online way. In our setting, since a new transcript guess is a two-bit modification of the last transcript guess, we can have Alice and Bob track a sequence of updates $U_A, U_B \in \{ 0, 1, \rewind, \bullet \}^*$ they have made to obtain their current transcript guess, where $\bullet$ is a placeholder update that simply means ``do nothing.'' Then, the question asker will send just the next two symbols of a tree code encoding of $U_A$ or $U_B$, which will take $O_\epsilon(1)$ bits per round. The receiver can then decode the entire history of received messages to determine the sequence of updates, which will allow them to determine the transcript being asked.

In our $O_\epsilon(|\pi_0|^2)$ protocol, we had the property that for Alice to successfully decode the asked transcript, she only needed to receive the last message (which contained the entire asked transcript) correctly. However, in a traditional tree code, even if Alice received the last message correctly, she cannot decode the message history if she received a high fraction (specifically more than half) of the previous messages incorrectly. In this paper, we present a new notion of \emph{sensitive} tree codes that in fact satisfy a stronger property, that for all but $\epsilon |w|$ indices $i$ where $w[i] = LTC(x)[i]$, it in fact holds that decoding $w[1:i]$ will \emph{uniquely} give $x[1:i]$. This essentially means that Alice only needs to receive the previous symbol of a sensitive tree code correctly to determine the entire message so far.\footnote{Sensitive tree codes can also be thought of as codes where the message can (usually) be decoded uniquely as long as the suffix distance to the original codeword is at most $1-\epsilon$. Previous results only guaranteed a message could be decoded correctly when the suffix distance was $\frac12-\epsilon$ to the original codeword; for example Lemma 2.3 in \cite{Gelles-survey}.}

Our notion of sensitive tree codes follows a similar construction as \emph{list tree codes}, introduced by Braverman and Efremenko~\cite{BravermanE14}. These are codes which guarantee that there is on average some constant number of ways to decode a random prefix of a string $w$. What we show is that this constant can actually be made $1$.

Still, we need answers to have message size $O_\epsilon(1)$ as well. To achieve this, we make the following modification to the answer format. Instead of sending $\ECC(T^*, \delta)$, which has size $O_\epsilon(|\pi_0|)$, a party who wishes to answer the transcript specified by the sequence of operations $U^*$ instead sends $\ECC(\sigma, \delta)$, where $\sigma$ is the last two symbols in the list tree code encoding of $(U^* || \bullet\bullet)$.

There is still one case where the new protocol is not analogous to the one from Section~\ref{sec:overview-pi^2}. In the protocol from Section~\ref{sec:overview-pi^2}, when Alice is asking the same transcript $T'$ that she is answering, she sends $\ECC(T',?)$ as a question. Bob will notice that $T'$ happens to be the same as the question he asked, and update with probability $1$. In some sense, this message gives Alice the benefits of both asking and answering a question. However, in the new setup, in order to ask a question, Alice has to send the last two symbols of the encoding of $U_A$, but in order to answer $U^*_B$ she has to send the last two symbols of $U^*_B$. The issue is that these symbols may not be the same, even if $U_A$ and $U^*_B$ correspond to the same complete transcript $T'$.

This leads us to define a new sort of online-updatable code, where if two histories correspond to the same transcript, even if the histories themselves are different, the next tree code encoding of a given edge is the same. This requires defining a code on a particular graph rather than on trees.

\subsection{Codes on Graphs}

Consider the rooted $|\Sigma_{in}|$-ary tree {\small\cursive{T}}. A sequence of symbols $\in \Sigma_{in}$ can be associated with a rooted path of {\small\cursive{T}} in the natural way. A sensitive tree code is then an assignment of symbols in $\Sigma_{out}$ to the edges of {\small\cursive{T}}. To encode a string $x \in \Sigma_{in}^k$, one simply traverses the corresponding rooted path and writes down the symbols seen. This gives an encoding $\in \Sigma_{out}^k$. 

The problem with using sensitive tree codes for our purposes is that Alice may have followed one path to get to the correct transcript $T_A = \cT$ while Bob followed another to get to $T_B = \cT$. Then, the next edge for Alice is different then the next edge for Bob, which means that one cannot hope to coincide sending the next symbol of one's own tree code with answering the other's.

Our key observation is that the encoding of the next symbol depends only on the transcript so far, not the full history of symbols.
So, we can actually coincide all nodes of {\small\cursive{T}} that lead to the same transcript. We define the following graph.

\paragraph{The Graph.} 

The graph $G$ that we will be interested in is defined as follows:
\begin{itemize}
\item 
    $G$ is a directed graph with vertices partitioned into layers $1,2,\ldots$. In the $i$'th layer, there is a vertex for each possible transcripts of length $\le i$. In particular, there is one vertex in the $0$'th layer, namely, the empty string.
\item 
    We set $\Sigma_{in} = \{ 0, 1, \rewind, \bullet \}$ to be the possible update instructions, where $\bullet$ means simply ``do nothing.'' Each vertex in the $i$'th layer has $4$ children in the $(i+1)$'th layer, corresponding to the $4$ resulting transcripts obtained by applying an instruction in $\Sigma_{in}$ to the vertex's associated transcript. 
\end{itemize}

Note that any sequence of updates $\in (\Sigma_{in})^*$ corresponds to a rooted path in $G$. Furthermore, any two equal length sequences of updates that result in the same transcript end at the same node.

\paragraph{The Code on $G$.}
We define a \emph{layered code} to be an assignment of elements of $\Sigma_{out}$ to the edges of $G$. Then, to encode $x \in (\Sigma_{in})^*$, one simply follows the path specified by $x$ and records the $|x|$ symbols seen on the edges.

We will use a specific layered code $\PikC$ that exhibits the same behavior as the sensitive tree codes we defined in Section~\ref{sec:overview-graph-codes}. We call these codes \emph{sensitive layered codes}. In particular, the property we want is that for all but $\epsilon |w|$ indices $i$ where $w[i] = \PikC(x)[i]$, decoding $w[1:i]$ gives a unique \emph{vertex} (i.e. transcript guess) equal to the vertex at the end of the rooted path specified by $x[1:i]$.

We will not go into depth how such to prove the existence of such a code here, but instead refer the reader to Section~\ref{sec:layered-graph-codes} for a comprehensive discussion. While much of our construction and proofs are motivated by the list tree codes of~\cite{BravermanE14}, we remark that there are several subtleties that need to be carefully addressed. 


\subsection{Boosting to Achieve Computational Efficiency}

Thus far, we have described how to obtain an interactive coding scheme that is resilient to $\frac16 - \epsilon$ error and has communication complexity linear in the size of the original protocol. Unfortunately, since decoding our sensitive layered code is inefficient (in fact, takes exponential time), this means that the computation needed by both parties is exponential in $|\pi_0|$. Thus, the final needed component is a way to make our scheme efficiently computable.

Over a large alphabet, an efficiently computable, positive rate scheme that is maximally error resilient was constructed by~\cite{GhaffariH13}. They obtained this efficient scheme in two steps: first by \emph{boosting} a known inefficient, exponential-time scheme~\cite{BravermanR11} to obtain an efficient protocol with a list-decoding guarantee, and second by applying a transformation that takes a list-decoding protocol to a unique-decoding protocol. We remark that this second transformation crucially relies on using a large alphabet and thus will not be permittable for us.

The boosted list-protocol is obtained as follows. First, they split up their original noiseless protocol into $\log^4 |\pi_0|$ size chunks. Then, they use their inefficient scheme to simulate the following noiseless subprotocol $O_\epsilon(\frac{|\pi_0|}{\log^4 |\pi_0|})$ times: 
\begin{itemize}
    \item Alice and Bob first find the longest transcript they have both simulated so far. This takes $O(\log^4 |\pi_0|)$ rounds. 
    \item Next, they run the next chunk of $\log^4 |\pi_0|$ rounds of the noiseless protocol.
\end{itemize}
Whenever a simulated subprotocol results in a completed transcript, that complete transcript obtains a vote. At the end, they show that as long as there was not too much corruption, the correct transcript must be one of the transcripts with the most votes (i.e. each party obtains a list of possible transcripts containing the correct one). Note that this results in a protocol with computational complexity $O_\epsilon(\frac{|\pi_0|}{\log^4 |\pi_0|}) \cdot \exp (\log^4 |\pi_0|) = \exp(\polylog |\pi_0|)$ time, which is considerably better than $\exp(|\pi_0|)$. Recursively boosting a second time gets the computational complexity down to $\poly(|\pi_0|)$. A third time reduces the computational complexity to $\tilde{O}_\epsilon(|\pi_0|)$.


\cite{GhaffariH13}'s second step is to apply a transformation that takes a list-decoding protocol to a unique decoding protocol, incurring a blowup in the alphabet size. Since we are working over a binary alphabet, we cannot afford to apply this same second transformation. Instead, we notice that our inefficient protocol has a property that we call \emph{scaling}. Essentially, this means that the amount of confidence Alice and Bob have in their final transcript guesses is directly related to the amount of corruption the adversary put in. More specifically, if the adversary corrupted $\frac16 - \rho$ of the communication ($\rho > 0$), then Alice and Bob end up with the correct transcript and are $\propto \rho$ confident in its correctness; and if the adversary corrupted $\frac16 + \rho$ of the communication, then Alice and Bob may end up with incorrect transcripts but they are only $\propto \rho$ confident. We can understand this as saying that $\frac16 - \rho$ corruption results in a net good confidence of $\rho$ (where $\rho$ can be positive or negative: $\rho < 0$ means that there was $\rho$ confidence in a bad transcript). 

This allows us to consider the same boosting transformation that~\cite{GhaffariH13} did, with the following caveat: whenever a simulated subprotocol results in a complete transcript, that transcript obtains a vote \emph{proportional to the confidence the parties have in the simulated protocol's correctness}. Then, if the adversary corrupts $< \frac16$ of the protocol, the net good votes (i.e. the number of votes for the correct transcript minus the total number for all incorrect transcripts) must be positive, so Alice and Bob can determine the correct transcript.


We elaborate more on our boosting transformation in Section~\ref{sec:boosting}.

\section{Preliminaries} \label{sec:prelims}


\paragraph{Notation.} In this work, we use the following notations.
\begin{itemize}
    \item The function $\Delta(x, y)$ represents the Hamming distance between $x$ and $y$.
    \item $x[i]$ denotes the $i$'th bit of a string $x \in \{ 0, 1 \}^*$.
    \item $x[i:j]$ denotes the $i\ldots j$'th bits of $x \in \{0,1\}^*$.
    \item $x||y$ denotes the string $x$ concatenated with the string $y$.
\end{itemize}

\subsection{Noise Resilient Interactive Communication} \label{sec:ip-def}

We formally define a non-adaptive interactive protocol and with error resilience. Our definition is for the binary alphabet $\{ 0, 1 \}$.

\begin{definition} [Non-Adaptive Interactive Coding Scheme] \label{def:ip}
    A two-party non-adaptive interactive coding scheme $\pi$ for a function $f(x, y) : \{ 0, 1 \}^n \times \{ 0, 1 \}^n \rightarrow \{ 0, 1 \}^o$ is an interactive protocol consisting of a fixed number of transmissions, denoted $|\pi|$. In each transmission, a single party fixed beforehand sends a single bit to the other party. At the end of the protocol, each party outputs a guess $\in \{ 0, 1 \}^o$.
    
    
    We say that $\pi$ is \emph{resilient to $\alpha$ fraction of adversarial errors with probability $p$} if the following holds. For all $x,y \in \{ 0, 1 \}^n$, and for all adversarial attacks consisting of at most $\alpha \cdot |\pi|$ errors, with probability $\ge p$ Alice and Bob both output $f(x,y)$ at the end of the protocol. 
\end{definition}

It is known that over a binary alphabet, one cannot achieve an error resilience greater than $\frac16$. 

\begin{theorem} [\cite{EfremenkoGH16}] \label{thm:EGH16-1/6}
    There exists a function $f(x, y)$ of Alice and Bob's inputs $x,y\in\{0,1\}^n$, such that any non-adaptive interactive protocol over the binary bit flip channel that computes $f(x, y)$ succeeds with probability at most $\frac12$ if a $\frac16$ fraction of the transmissions are corrupted.
\end{theorem}

\section{Boosting: Obtaining Computational Efficiency}\label{sec:boosting}

In this section, we show how to boost the computational efficiency of a scheme. Our boosted protocol draws inspiration from the list-decoding boosting scheme of~\cite{GhaffariH13}, which drew ideas from~\cite{BrakerskiK12}. We begin by recalling the necessary setup from~\cite{GhaffariH13}.

\subsection{The Simulation Paradigm of~\cite{GhaffariH13,BrakerskiK12}}

Assume that $\pi_0$ is an alternating binary protocol of length $n_0$ (any binary protocol can be made alternating by increasing the communication by at most a factor of $2$). We can view $\pi_0$ as a \emph{protocol tree} $\mathbb{T}$, in which the edges at odd levels correspond to Alice's messages and the edges at even levels correspond to Bob's messages. For any input $x$, $\pi_0$ defines a subset $S_A$ of edges at the odd levels corresponding to Alice's possible responses, and similarly, for any input $y$, $\pi_0$ defines a subset $S_B$ of edges at the even levels corresponding to Bob's possible messages. Note that for any $(x, y)$, $S_A \cup S_B$ defines a unique rooted path $\cT$ corresponding to the noiseless protocol $\pi_0(x,y)$. The goal is for both Alice and Bob to determine $\cT$.

To do this, Alice and Bob each keep track of a set of edges $\cE_A$ and $\cE_B$. Initially both sets are empty. In each of many iterations, Alice (resp. Bob) will add some edges to $\cE_A$ (resp. $\cE_B$) extending some existing path in $\cE_A$ (resp. $\cE_B$). We remark that any new edges Alice adds must be consistent with her own behavior on her input $x$, i.e. she never adds an edge in an odd layer that does not belong to $S_A$. The same holds for Bob. It thus holds that at any point the unique longest rooted path in both $\cE_A$ and $\cE_B$ is a prefix of $\cT$.

The process by which Alice and Bob add edges to their respective set in each iteration is as follows. They first run a subprotocol to determine their longest common rooted path. Then, they run the next $\log^4 n_0$ rounds of the noiseless protocol. They perform both these steps under a single error-resilient simulation. The idea is that every time not too many errors have happened in an iteration, both Alice and Bob add $\log^4 n_0$ edges to the correct path corresponding to $\cT$.

If the longest common rooted path is a path from the root to a leaf, then Alice and Bob instead add some weight to that leaf. Over the course of many iterations, the hope is that the leaf with the largest weight at the end of the protocol should correspond to $\cT$. We remark that~\cite{GhaffariH13} showed a list-guarantee assuming not too many errors occurred: at the end of this procedure, Alice and Bob will each have a small list of leaves each containing the true leaf corresponding to $\cT$. (They then need to run this procedure many times in parallel with sending an error correcting code in order for both parties to narrow down the correct transcript, resulting in an alphabet blowup.) For us, we will show that if our inefficient simulation has a property known as \emph{scaling} (see Definition~\ref{def:scaling}), then at the end of this procedure Alice and Bob will each have narrowed down to a \emph{unique} leaf, precisely, the leaf corresponding to $\cT$, provided not too many errors occurred. 



\paragraph{The Tree-Intersection Problem.}
The problem of finding their longest shared path is called the \emph{tree-intersection problem}. Precisely, assuming Alice and Bob have sets of edges $\cE_A$ and $\cE_B$ respectively each forming a rooted tree under the promise that $\cE_A \cap \cE_B$ is a rooted path, the problem is for Alice and Bob to recover this rooted path using as little communication and computation as possible. 

In \cite{GhaffariH13}, they give a data structure for $\cE_A$ and $\cE_B$ that optimizes the computational complexity of a protocol solving the tree-intersection problem.

\begin{theorem}\cite{GhaffariH13} \label{thm:data_structure}
    There is an incremental data structure that maintains a rooted subtree of the rooted infinite binary tree under edge additions with amortized computational complexity of $\tilde{O}(1)$ time per edge addition. Furthermore, for any $c = \Omega(1)$ and given two trees of maximum size $n$ maintained by such a data structure, there is a tree-intersection protocol that uses $100 c \log^4 n$ rounds of communication over a noiseless binary channel, $O(c \log^4 n)$ bits of randomness, and $\tilde{O}(1)$ computation steps to solve the tree intersection problem, that is, find the intersection path with failure probability at most $2^{-c \log^4 n}$. 
\end{theorem}

\subsection{Scaling Schemes}

We now define precisely what we mean by a \emph{scaling} scheme. Intuitively, a scaling scheme is a scheme in which Alice and Bob output a \emph{confidence} in addition to a transcript. This confidence should give a bound on the total error in the protocol. For instance, if there is no corruption, then Alice and Bob should output the correct transcript with large confidence. If there is some corruption, then Alice and Bob should output the correct transcript with smaller confidence. If there is too much corruption, then Alice and Bob may output an incorrect transcript, but their confidence cannot exceed a certain quantity specified by the amount of error that occurred (i.e. if the adversary wishes Alice and Bob to be more confident in an incorrect transcript, she must corrupt more of the protocol).

\begin{definition}[$(\rho, \epsilon, \mu_\epsilon)$-Scaling Schemes] \label{def:scaling}
    A scheme for simulating a noiseless protocol of length $n$ is \emph{$(\rho, \epsilon, \mu_\epsilon)$-scaling} if, at the end of the protocol, Alice and Bob output guesses $T_A$ and $T_B$ for the noiseless transcript $\cT$ along with confidences $c_A, c_B \in [0, 1]$, with the following guarantees:
    \begin{itemize}
    \item {\bf Consistency:}
        All of Alice's messages in $T_A$ are consistent with her behavior in $\pi_0$ on input $x$. Similarly, all of Bob's messages in $T_B$ are consistent with his behavior in $\pi_0$ on input $y$.
    \item {\bf Scaling 1:}
        If a $\delta < (1-\epsilon) \cdot \rho$ fraction of the scheme was corrupted, then
        \[
            \Pr \left[ T_A = T_B = \cT ~\wedge~ c_A, c_B \ge 1 - \frac{\delta}{\rho} - \epsilon \right] \ge 1 - \mu_\epsilon(n).
        \]
    \item {\bf Scaling 2:}
        If $\delta \ge (1 - \epsilon) \cdot \rho$ fraction of the scheme was corrupted, then 
        \[
            \Pr \left[ \left( T_A \not= \cT ~\wedge~ c_A > \frac{\delta}{\rho} - 1 + \epsilon \right) \vee \left( T_B \not= \cT ~\wedge~ c_B > \frac{\delta}{\rho} - 1 + \epsilon \right) \right] \le \mu_\epsilon(n).
        \]
    \end{itemize}
\end{definition}

\subsection{Boosting}

\protocol{Boosting}{boosting}{
    Let $\cP'$ be a $(\rho, \epsilon, \mu_\epsilon)$-scaling scheme that simulates noiseless protocols of length $n'$ by a protocol of length $r_\epsilon(n')$ that has computational complexity $T_\epsilon(n')$. Choose $C_\epsilon \ge 100/\epsilon + 1$.
    
    For a protocol $\pi_0$ that has length $n_0$, and on inputs $(x, y)$, Alice and Bob run the following scheme:
    \begin{enumerate}
    \item 
        Alice and Bob each keep track of a list $\cE_A, \cE_B \subseteq \mathbb{T}$ of edges they have simulated so far, using the data structure from~\ref{thm:data_structure}. Initially, $\cE_A, \cE_B = \emptyset$. They also each keep track of a dictionary\footnote{Roughly, a dictionary is implemented by a hash table.} $\cL_A, \cL_B$ of leaves, i.e. full transcripts $T$ of $\mathbb{T}$, mapping to $\bbR_{\ge 0}$. Initially, for any full transcript $T$ of $\mathbb{T}$, $\cL_A[T] = \cL_B[T] = 0$. 
    \item 
        For $i = 1, \dots, \frac{n_0}{\epsilon \log^4 n_0} =: \beta$, they use $\cP'$ to simulate the following $n' = C_\epsilon \cdot \log^4 n_0$ round noiseless protocol:
        \begin{enumerate}
        \item 
            Alice and Bob run the tree-intersection protocol given in Theorem~\ref{thm:data_structure}, using $(C_\epsilon - 1) \log^4 n_0$ rounds and $\tilde{O}(1)$ computation steps. At the end, with probability $1 - 2^{-((C_\epsilon - 1)/100) \cdot \log^4 n_0} \ge 1 - 2^{-\log^4 n_0/\epsilon}$, the two parties have determined the common rooted path $p = \cE_A \cap \cE_B$.
        \item 
            After Alice and Bob have determined a common path $p$, they fix $p$ to be the transcript prefix of $\pi_0$ so far and run the next $\log^4 n_0$ rounds of $\pi_0$. (If there are fewer than $\log^4 n_0$ rounds in $\pi_0$ remaining after $p$, they treat the remaining rounds as sending all $0$'s.)
        \end{enumerate}
        At the end of the simulation, Alice has determined a transcript prefix $p_A \subseteq \cE_A$ along with up to $\log^4 n_0$ subsequent edges extending $p_A$. She also has a confidence $c_A \in [0, 1]$. She adds the $\le \log^4 n_0$ edges to $\cE_A$ (ignoring duplicates). Further, if $p_A$ is a complete transcript of length $n_0$, she adds $c_A$ to $\cL_A[p_A]$. Bob does the same.
    \item 
        At the end of the protocol, let $T_A = \arg\max_{p} \cL_A[p]$ be the transcript with the highest weight in $\cL_A$, and let $w_A = \cL_A[T_A]$. Also, let $w^c_A = \sum_{p \not= T_A} \cL_A[p]$ be the total weight assigned to all the other leaves excluding $T_A$. Then, Alice outputs $T_A$, along with confidence $c_A = \frac{w_A - w^c_A}{\beta}$.
        
        Similarly, Bob outputs the transcript $T_B = \arg\max_{p} \cL_B[p]$ and confidence $c_B = \frac{w_B - w^c_B}{\beta}$, where $w_B = \cL_B[T_B]$ and $w^c_B = \sum_{p \not= T_B} \cL_B[p]$ is the total weight on all the other leaves excluding $T_B$.
    \end{enumerate}
}

\begin{theorem} \label{thm:boosting}
    Let $\epsilon < 0.25$ and $C_\epsilon \ge 100/\epsilon + 1$. Assume a $(\rho, \epsilon, \mu_\epsilon)$-scaling scheme that simulates noiseless protocols of length $n$ with communication complexity $r_\epsilon(n)$ and computational complexity $T_\epsilon(n)$. Then, the protocol given in Protocol~\ref{prot:boosting} is a $(\rho, 4\epsilon, e^{-\epsilon n_0 / 10 \log^4 n_0})$-scaling scheme for noiseless protocols of length $n_0$ that has communication complexity $\frac{n_0}{\epsilon \log^4 n_0} \cdot r_{\epsilon}(C_\epsilon \cdot \log^4 n_0)$ and computational complexity $\tilde{O_{\epsilon}}(n_0) \cdot T_\epsilon(C_\epsilon \log^4 n_0)$, assuming that $\mu_\epsilon(C_\epsilon \log^4 n_0) < \frac\epsilon4$.
\end{theorem}

\begin{proof}
    Clearly, the communication complexity in Protocol~\ref{prot:boosting} is $\frac{n_0}{\epsilon \log^4 n_0} \cdot r_\epsilon (C_\epsilon \log^4 n_0)$. As for the computational complexity, note that in each iteration, Alice needs to do $T_\epsilon(C_\epsilon \log^4 n_0)$ computations to obtain a transcript $T'$ and a confidence $c'$. She may further have to update $\cL_A[T]$ with the confidence $c'$, for some complete transcript $T$, which can be done in amortized $O(\log L)$ time since a dictionary is roughly implemented by a hash table, where $L$ is an upper bound on the size of $\cL_A$. Finally, at the end of the protocol, she can determine $T_A, w_A, w^c_A$ by making a linear pass through $\cL_A$. Thus, the total computational complexity is $\beta \cdot (T_\epsilon(C_\epsilon \log^4 n_0) + O(\log L)) + \tilde{O}(L)$. Since $L \le \beta$, which follows from the fact that Alice makes at most one value of $\cL_A[p]$ nonzero in each iteration, the total computational complexity is $\tilde{O}(\beta) \cdot T_\epsilon(C_\epsilon \log^4 n_0)$ which is at most $\tilde{O_{\epsilon}}(n_0) \cdot T_\epsilon(C_\epsilon \log^4 n_0)$.
    
    We will now show that our scheme is $(\rho, 4\epsilon, e^{-\epsilon n_0 / 10 \log^4 n_0})$-scaling. First, the consistency property follows because each of the protocols in the $\beta$ iterations are consistent: Alice and Bob only add edges to $\cE_A, \cE_B$ that are consistent with their own input, so only transcripts consistent with their own input can gain weight in $\cL_A, \cL_B$. The rest of this proof will show the scaling properties.
    
    Let $\delta_1, \dots, \delta_\beta$ be the fractional amount of corruption in each of the $\beta$ simulations, so that the total fractional amount of error is $\delta = \frac1\beta \sum_{i = 1}^\beta \delta_i$. Let $T'_{A,1}, \dots, T'_{A,\beta}$ and $c'_{A,1}, \dots, c'_{A,\beta}$ (resp. $T'_{B,1}, \dots, T'_{B,\beta}$ and $c'_{B,1}, \dots, c'_{B,\beta}$) be the transcripts and confidences Alice (resp. Bob) has at the end of each of the $\beta$ simulations.
    
    Denote by $E_i(T'_i)$ denote the event that in the transcript $T'_i$, Alice and Bob correctly determine their longest shared path $\cE_A \cap \cE_B$ and extend it by $\log^4 n_0$ bits (or send $0$'s once the total transcript exceeds length $n_0$).
    
    \begin{lemma} \label{lemma:boosting-iteration-scaling}
        The following holds for the simulation in the $i$'th iteration: 
        \begin{itemize} 
        \item 
            If there are at most $\delta_i < (1 - \epsilon) \cdot \rho$ errors, then 
            \[
                \Pr \left[ E_i(T'_{A,i}) ~\wedge~ E_i(T'_{B,i}) ~\wedge~ c'_{A,i}, c'_{B,i} \ge 1 - \frac{\delta_i}{\rho} - \epsilon \right]
                \ge 1 - \mu_\epsilon(C_\epsilon \cdot \log^4 n_0) - 2^{-c \log^4 n_0}.
            \]
        \item 
            If there are at least $\delta_i \ge (1 - \epsilon) \cdot \rho$ errors, then 
            \begin{align*}
                \Pr \left[ \left( \neg E_i(T'_{A,i}) ~\wedge~ c'_A > \frac{\delta_i}{\rho} - 1 + \epsilon \right) \vee \left( \neg E_i(T'_{B,i}) ~\wedge~ c'_B > \frac{\delta_i}{\rho} - 1 + \epsilon \right) \right] \\
                \le \mu_\epsilon(C_\epsilon \cdot \log^4 n_0) + 2^{-c \log^4 n_0}.
            \end{align*}
        \end{itemize} 
    \end{lemma}
    
    \begin{proof}
        First, suppose that $\delta_i < (1 - \epsilon) \cdot \rho$. Let $T^*_i$ denote the noiseless protocol in the $i$'th simulation. Note that with probability $e^{\log^4 n_0 / \epsilon}$, $T^*_i$ may not correctly determine Alice and Bob's longest shared path. In particular,
        \begin{align*}
            \Pr&~ \left[ 
                \neg \left( E_i(T'_{A,i}) ~\wedge~ E_i(T'_{B,i}) ~\wedge~ c'_{A,i}, c'_{B,i} \ge 1 - \frac{\delta_i}{\rho} - \epsilon \right)
            \right]\\ 
            \le&~ \Pr \left[ 
                \neg E_i(T^*_i)
            \right]
            + \Pr \left[ 
                \neg \left( T'_{A,i} = T'_{B,i} = T^*_i ~\wedge~ c'_{A,i}, c'_{B,i} \ge 1 - \frac{\delta_i}{\rho} - \epsilon \right) 
            \right] \\
            \le&~ 2^{\log^4 n_0 / \epsilon} + \mu_\epsilon(C_\epsilon \cdot \log^4 n_0)
        \end{align*}
        by Theorem~\ref{thm:data_structure} and Definition~\ref{def:scaling}.
        
        On the other hand, if $\delta_i \ge (1 - \epsilon) \cdot \rho$, it holds that 
        \begin{align*}
            &~ \Pr \left[ \left( \neg E_i(T'_{A,i}) ~\wedge~ c'_A > \frac{\delta_i}{\rho} - 1 + \epsilon \right) \vee \left( \neg E_i(T'_{B,i}) ~\wedge~ c'_B > \frac{\delta_i}{\rho} - 1 + \epsilon \right) \right] \\
            &\le \Pr [ \neg E_i(T^*_i) ] + \Pr \left[ \left( T'_{A,i} \not= T^*_i ~\wedge~ c'_A > \frac{\delta_i}{\rho} - 1 + \epsilon \right) \vee \left( T'_{B,i} \not= T^*_i ~\wedge~ c'_B > \frac{\delta_i}{\rho} - 1 + \epsilon \right) \right] \\
            &\le 2^{\log^4 n_0/\epsilon} + \mu_\epsilon(C_\epsilon \cdot \log^4 n_0),
        \end{align*}
    where the second line follows from considering the cases where $\neg E_i(T^*_i)$ and $E_i(T^*_i)$, and the third line follows from Theorem~\ref{thm:data_structure} and Definition~\ref{def:scaling}.
    \end{proof}
    
    Let $I \subseteq [\beta]$ denote the iterations in which $< (1-\epsilon) \cdot \rho$ of the scheme was corrupted. 
    
    \begin{lemma} \label{lemma:GammaLambda}
        With probability $1 - e^{-\epsilon^2 \beta / 10}$, for all except at most $\epsilon \cdot \beta$ values of $i \in [\beta]$, it holds that either:
        \begin{enumerate}[label={(\arabic*)}]
            \item $i \in I$ and $E_i(T'_{A,i}) ~\wedge~ E_i(T'_{B,i}) ~\wedge~ c'_{A,i}, c'_{B,i} \ge 1 - \frac{\delta_i}{\rho} - \epsilon$, \label{item:condition-I}
            \item $i \in [\beta] \backslash I$ and $\left( E_i(T'_{A,i}) \vee c'_A \le \frac{\delta_i}{\rho} - 1 + \epsilon \right) \wedge \left( E_i(T'_{B,i}) \vee c'_B \le \frac{\delta_i}{\rho} - 1 + \epsilon \right)$. \label{item:condition-Ic}
        \end{enumerate}
    \end{lemma}
    
    \begin{proof}
        By Lemma~\ref{lemma:boosting-iteration-scaling}, one of the two conditions holds for every $i \in [\beta]$ with probability at least $1 - 2^{-\log^4 n_0 / \epsilon} - \mu_\epsilon (C_\epsilon \cdot \log^4 n_0)$. This means that the expected number of $i$ satisfying one of the two conditions is $\varpi \ge (1 - 2^{-\log^4 n_0 / \epsilon} - \mu_\epsilon(C_\epsilon \log^4 n_0)) \cdot \beta$. 
        
        Let $X$ denote the number of $i \in [\beta]$ satisfying one of the two conditions. By Chernoff, 
        \begin{align*}
            \Pr [X < (1 - \epsilon) \cdot \beta] 
            \le \Pr [X < (1 - \epsilon/2) \cdot \varpi] 
            \le e^{-\epsilon^2 \varpi / 8}
            \le e^{-\epsilon^2 \beta / 10},
        \end{align*}
        where the first and last inequalities follow from the fact that $2^{-\log^4 n_0 / \epsilon} + \mu_\epsilon(C_\epsilon \log^4 n_0) \le 2^{-1/\epsilon} + \mu_\epsilon(C_\epsilon \log^4 n_0) < \frac\epsilon4 + \frac\epsilon4 = \frac\epsilon2$, so $(1 - \epsilon/2) \cdot \beta < \varpi$. In particular, the first inequality follows from $(1 - \epsilon) \beta < (1 - \epsilon/2)^2 \beta < (1 - \epsilon/2) \varpi$, and the last inequality follows from $0.8 \beta < (1 - \epsilon/2) \beta < \varpi$.
    \end{proof}
    
    Let $\Gamma \subseteq I$ be the set of all $i$ satisfying~\ref{item:condition-I}, and let $\Lambda \subseteq [\beta] \backslash I$ be the set of all $i$ satisfying~\ref{item:condition-Ic}. Note that after the first $\frac{n_0}{\log^4 n_0}$ iterations in $\Gamma$, Alice and Bob are both guaranteed to have all edges in the correct transcript $\cT$ in their edge lists $\cE_A$ and $\cE_B$. After that point, in every iteration in $\Gamma$, Alice and Bob both determine the correct transcript $\cT = \cE_A \cap \cE_B$ and add $c'_{A,i}$ (resp. $c'_{B,i}$) to $\cL_A[\cT]$ (resp. $\cL_B[\cT]$). This means that at the end of the protocol,
    \[
        \cL_A[\cT] \ge \sum_{i \in \Gamma} c'_{A,i} - \frac{n_0}{\log^4 n_0}
        \ge (1-\epsilon) \cdot |\Gamma| - \frac1\rho \cdot \sum_{i \in \Gamma} \delta_i - \frac{n_0}{\log^4 n_0},
    \]
    and similarly
    \[
        \cL_B[\cT] \ge (1-\epsilon) \cdot |\Gamma| - \frac1\rho \cdot \sum_{i \in \Gamma} \delta_i - \frac{n_0}{\log^4 n_0}.
    \]
    
    Meanwhile, for each iteration in $\Lambda$, a weight of at most $c'_{A,i}$ (resp. $c'_{B,i}$) is added to a wrong leaf. Furthermore, a weight of at most $1$ is added to a wrong leaf for each iteration in $[\beta] \backslash (\Gamma \cup \Lambda)$, which by Lemma~\ref{lemma:GammaLambda} has size at most $\epsilon \beta$ with probability $1 - e^{-\epsilon^2 \beta / 10}$. Thus, with probability $1 - e^{-\epsilon^2 \beta / 10}$, the total weight on all the wrong leaves in Alice's tree is at most
    \[
        \le \sum_{i \in \Lambda} c'_{A,i} \cdot \mathbbm{1}[T'_{A,i} \not= T^*_i]
        + \sum_{i \in [\beta] \backslash (\Gamma \cup \Lambda)} 1
        \le \frac1\rho \cdot \sum_{i \in \Lambda} \delta_i - (1 - \epsilon) \cdot |\Lambda| + \epsilon \beta,
    \]
    and simultaneously the total weight on all the wrong leaves in Bob's tree is at most
    \[
        \le \sum_{i \in \Lambda} c'_{B,i} \cdot \mathbbm{1}[T'_{B,i} \not= T^*_i]
        + \sum_{i \in [\beta] \backslash (\Gamma \cup \Lambda)} 1
        \le \frac1\rho \cdot \sum_{i \in \Lambda} \delta_i - (1 - \epsilon) \cdot |\Lambda| + \epsilon \beta.
    \]
    
    Then, with probability $1 - e^{-\epsilon^2 \beta / 10}$, the difference between the weight on the correct leaf and the combined weight on all the wrong leaves, for both Alice and Bob, is 
    \begin{align*}
        \cL_A[\cT] &- \sum_{T \not= \cT} \cL_A[T] ~\text{(resp.}~\cL_B[\cT] - \sum_{T \not= \cT} \cL_B[T]\text{)}\\
        &\ge \left[ (1 - \epsilon) \cdot |\Gamma| - \frac1\rho \cdot \sum_{i \in \Gamma} \delta_i - \frac{n_0}{\log^4 n_0} \right]
        - \left[ \frac1\rho \cdot \sum_{i \in \Lambda} \delta_i - (1 - \epsilon) \cdot |\Lambda| + \epsilon \beta \right] \\
        &= (1 - \epsilon) \cdot (|\Gamma| + |\Lambda|) - \epsilon \beta - \frac1\rho \cdot \sum_{i \in \Gamma \cup \Lambda} \delta_i - \frac{n_0}{\log^4 n_0} \\
        &\ge (1 - \epsilon) \cdot (\beta - \epsilon \beta) - \epsilon \beta - \frac{\delta \beta}{\rho} - \epsilon \beta \\
        &\ge \left( 1 - \frac\delta\rho - 4\epsilon \right) \cdot \beta, \numberthis \label{eqn:difference-weights}
    \end{align*}
    where we used that $\beta = \frac{n_0}{\epsilon \log^4 n_0}$ and that $\sum_{i \in \Gamma \cup \Lambda} \delta_i \le \sum_{i \in [\beta]} \delta_i = \delta \beta$. 
    
    In particular, if $\delta < (1 - \frac\delta\rho - 4\epsilon) \cdot \rho$, then with probability $1 - e^{-\epsilon^2 \beta / 10}$, both Alice and Bob output $T_A = T_B = \cT$ and confidence $c_A, c_B \ge 1 - \frac\delta\rho - 4\epsilon$.
    
    On the other hand, Equation~\ref{eqn:difference-weights} tells us that with probability $1 - e^{-\epsilon^2 \beta/10}$, for both Alice and Bob, for \emph{any} incorrect leaf $T_0$, the total weight on $T_0$ minus the combined weight on all the other leaves is at most 
    \[
        \le \left( \frac\delta\rho - 1 + 4\epsilon \right) \cdot \beta,
    \]
    since $\cL_A[T_0] \le \sum_{T \not= \cT} \cL_{A}[T]$, and $\sum_{T \not= T_0} \cL_A[T] \ge \cL_A[\cT]$ (and same for Bob). Thus, in the case that $\delta > (1 - \frac\delta\rho - 4\epsilon) \cdot \rho$ of the entire protocol is corrupted, it holds with probability $1 - e^{-\epsilon^2 \beta / 10}$ that either $T_A = \cT$, or $T_A \not= \cT$ and $c_A \le \frac\delta\rho - 1 + 4\epsilon$, and same for Bob.
    
    It follows that Protocol~\ref{prot:boosting} is $(\rho, 4\epsilon, e^{-\epsilon^2 \beta/10}) = (\rho, 4\epsilon, e^{-\epsilon n_0 / 10 \log^4 n_0})$-scaling.

\end{proof}

\section{Layered Codes} \label{sec:layered-graph-codes}
In this section, we introduce \emph{sensitive layered codes}, which are a generalization and strengthening of list tree codes to codes on layered graphs. List tree codes were first introduced in \cite{BravermanE14} as an analogue of list-decodable error correcting codes for the tree code setting. Sensitive layered codes are instead defined on certain graphs, and have list size $1$ for most locations.

We first define suffix distance.

\begin{definition}[Suffix Distance]
    For two strings $x, y \in \Sigma^n$, we define the suffix distance as follows:
    \[
        \Delta_{sfx}(x, y) = \max_{0 \le i \le n-1} \frac{\Delta(x[i+1:n], y[i+1:n])}{n-i}.
    \]
\end{definition}

\subsection{Layered Codes}\label{sec:graph-codes}



\begin{definition}[Layered Graph Over An Alphabet] \label{def:layered-graph}
    Let $\Sigma$ be an alphabet. A \emph{layered graph over  $\Sigma$ of depth $n$} is a directed graph $G$ that satisfies the following properties: 
    \begin{itemize}
        \item The vertices of $G$ can be split up into layers $0,1,\dots, n$. There is exactly one vertex in layer $0$.
        \item Each vertex in layer $i < n$ has out-degree exactly $|\Sigma|$: it has $|\Sigma|$ children in layer $i+1$, where the $|\Sigma|$ out-edges are associated with not necessarily distinct elements of $\Sigma$.
    \end{itemize}
\end{definition}


If $G$ is a layered graph over $\Sigma_{in}$ of depth $n$, note that any path $p$ in $G$ from the root node to a vertex in layer $i$ can be associated with a string $\in \Sigma_{in}^i$. Likewise, any string $\in \Sigma_{in}^i$ corresponds to a unique path in $G$ from the root node to a vertex in layer $i$. We will interchangeably refer to the path $p$ or the associated string $\in \Sigma_{in}^i$. Furthermore, for any string $p \in \Sigma_{in}^i$, we use $v(p)$ to denote the vertex at the end of $p$.

\begin{definition}[Layered Code]
    Let $G$ be a layered graph over $\Sigma_{in}$ of depth $n$. A \emph{layered code} $\PikC$ of $G$ with the alphabet $\Sigma_{out}$ is an assignment of elements of $\Sigma_{out}$ to the edges of $G$. We refer to such an assignment as a $(G, \Sigma_{out})$-code.
\end{definition}

For any subgraph $H\subseteq G$, we define $\PikC(H)$ to be the subgraph $H$ inheriting labels from $\PikC$. Specifically, for a rooted path $p \in \Sigma_{in}^i$, $\PikC(p)\in \Sigma_{out}^i$ is the string of $i$ labels of the edges in $p$. 

\subsection{Prefix Trees}

For any $(G,\Sigma_{out})$-code, any $\epsilon$, and any word $w \in \Sigma_{in}^n$, let the list $L_i(\PikC, w, \epsilon)$ be the list of nodes in layer $i$ that are the endpoint of at least one path whose encoding under $\PikC$ is close to the prefix of $w$ of length $i$ in their suffix distance. That is, 
\[
    L_i(\PikC, w, \epsilon) = \{ v(p) : p \in \Sigma_{in}^i ~\text{s.t.}~ \Delta_{sfx}(\PikC(p), w[1:i]) < 1 - \epsilon \}.
\]
We also write $L(\PikC, w, \epsilon) = \cup_{i=1}^n L_i(\PikC, w, \epsilon)$. 

Consider a subset $S \subseteq L(\PikC, w, \epsilon)$. For each $v \in S$, we pick a path $p$ from the root to $v$ satisfying $\Delta_{sfx} ( \PikC(p), w[1:|p|]) < 1 - \epsilon$. If these paths form a rooted tree, we call their union a \emph{prefix tree} of $S$. We denote by $\PT(\PikC, w, \epsilon)$ the set of all prefix trees of all subsets of $L(\PikC, w, \epsilon)$. 

\begin{lemma} \label{lemma:prefixtree-S}
    Fix $w \in \Sigma_{out}^n$ and $\epsilon > 0$. For any subset $S \subseteq L(\PikC, w, \epsilon)$, there is a prefix tree of $S$.
\end{lemma}

\begin{proof}
    For a path $q$ of length $k$, we define the \emph{deficit} of $q$, denoted $\mathsf{deficit}(q)$, to be $\max_{0 \le j < k} \left[ \Delta(\PikC(q)[j+1:k], w[j+1:k]) - (1 - \epsilon) \cdot (k-j) \right]$. For a path $p$ of length $i$, we say that the \emph{excess} of $p$ at $k \le i$ is $(1 - \epsilon) \cdot (i - k) - \Delta(\PikC(p)[k+1:i], w[k+1:i])$, denoted $\mathsf{excess}_k(p)$. Note that for any path $p$ for which $v(p) \in L_i(\PikC, w, \epsilon)$, it holds that $\mathsf{excess}_k(p) > 0$ for any $k \le i$. 
    
    Furthermore, we claim that for any $p \in \Sigma_{in}^i$ such that $v(p) \in L_i$, letting $p'$ denote the path obtained by replacing the first $k$ edges by $q \in \Sigma_{in}^k$, we have that $\Delta_{sfx}(\PikC(p'), w[1:i]) < 1 - \epsilon$ iff $\mathsf{deficit}(q) < \mathsf{excess}_k(p)$. To see this, we can write 
    \begin{align*}
        \Delta_{sfx} (\PikC(p'), w[1:i])
        = \max \left\{ 
        \begin{aligned}
            \Delta_{sfx} ( \PikC(p)[k+1:i], w[k+1:i] ), \\ \max_{0\le j < k}\frac{\Delta(\PikC(p)[k+1:i], w[k+1:i]) + \Delta(q[j+1:k], w[j+1:k]))}{i-j}
        \end{aligned} \right\}.
    \end{align*}
    Note that $\Delta_{sfx} ( \PikC(p)[k+1:i], w[k+1:i] ) < 1 - \epsilon$ because $v(p) \in L_i$. Thus, $\Delta_{sfx}(\PikC(p'), w[1:i]) < 1 - \epsilon$ iff 
    \[
        \Delta(\PikC(p)[k+1:i], w[k+1:i]) + \Delta(q[j+1:k], w[j+1:k])) < (1-\epsilon) \cdot (i-j)
    \]
    for all $0 \le j < k$, or equivalently,
    \[
        \mathsf{deficit}(q) < \mathsf{excess}_k(p).
    \]
    
    Now, given a selection of paths $\{ p(v) \}_{v \in S}$, where $p(v)$ connects the root to $v$, for each $k \in [n]$ define $\Lambda_k(p)$ to be the set of vertices $y \in G$ in layer $k$ such that there are two paths $p(v)$ and $p(v')$, where $v \not= v' \in S$, for which $v(p(v)[1:k]) = v(p(v')[1:k]) = y$ but $p(v)[1:k] \not= p(v')[1:k]$. We define $\Psi(p)$ to be $(k_{max}, |\Lambda_{k_{max}}(p)|)$, with the lexicographical ordering, where $k_{max}$ is the largest layer $k$ for which $\Lambda_k(p)$ is nonempty. 
    
    In order to construct a prefix tree of $S$, we begin by choosing a path $p(v)$ from the root to $v$ for each $v \in S$. Next, we perform an operation to $p$ that decreases $\Psi(p)$, while preserving that $p$ satisfies $\Delta_{sfx}(\PikC(p(v)), w[1:|p(v)|]) < 1 - \epsilon$ for all $v \in S$. The operation we perform is as follows: Choose $y_{max} \in \Lambda_{k_{max}}(p)$. Furthermore, let $v_1, \dots, v_m \in S$ be such that $v(p(v_\iota)[1:k]) = y_{max}$. Define $q_\iota := p(v_\iota)[1:k]$ for each $\iota \in [m]$. Let $\hat{\iota} = \arg\min_{\iota \in [m]} \mathsf{deficit}(q_\iota)$, and let $q = q_{\hat{\iota}}$. Then, for each $\iota \in [m]$, we replace $p(v_\iota)$ with the path $p'(v_\iota) = q || p(v_\iota)[k+1:|p(v_\iota)|]$. Since $\mathsf{deficit}(q) \le \mathsf{deficit}(q_\iota)$, it holds that $\Delta_{sfx}(\PikC(p'(v_\iota)), w[1:|p'(v_\iota)|]) < 1-\epsilon$ for all $\iota \in [m]$. (For all other $v \in S$ where $p(v)$ doesn't pass through $y_{max}$, we define $p'(v) = p(v)$.)
    
    Note that $\Lambda_k(p')$ where $k > k_{max}$ must still be empty, as we have only altered edges in layers at most $k_{max}$. Furthermore, $|\Lambda_{k_{max}}(p')|$ is strictly less than $|\Lambda_{k_{max}}(p)|$, since we have replaced paths going through $y_{max}$ with paths going through $y_{max}$ so no new intersections in layer $k_{max}$ were created, and we have removed $y_{max}$ from $\Lambda_{k_{max}}(p)$. Thus, $\Psi(p') < \Psi(p)$. Also note that as long as $\Psi(p) > (0,0)$, we can continue this operation, so eventually $\Psi(p) = (0,0)$, at which point the union of $p(v), v \in S$ is a tree.
\end{proof}

For a subgraph $H$ of $G$ of depth at most $|w|$, we denote by $w(H)$ the graph where we write $w[i]$ on all edges at depth $i$. For a $(G, \Sigma_{out})$-code $\PikC$, recall that $\PikC(H)$ is the subgraph $H$ inheriting labels from $\PikC$. For two labelings $w$ and $\PikC$ of a subgraph $H$, we define $agr(w(H), \PikC(H))$ to be the number of edges of $H$ for which the labels are the same.

\begin{lemma} \label{lemma:agr>epPT}
    For any $w \in \Sigma_{out}^n$ and $\epsilon > 0$, and for any $PT \in \PT(\PikC, w, \epsilon)$, 
    \[
        agr(\PikC(PT), w(PT)) > \epsilon |PT|.
    \]
\end{lemma}

\begin{proof}
    First, note that by definition of $L(\PikC, w, \epsilon)$, for any path $p$ ending at $v \in L(\PikC, w, \epsilon)$ and not necessarily starting at the root, it holds that $agr(\PikC(p), w(p)) > \epsilon |p|$. We call this Property A.
    
    We prove the lemma by induction on the number of leaves. If $PT$ has only $1$ leaf, then it is a path from root to leaf, and by Property A, $agr(\PikC(PT), w(PT)) > \epsilon |PT|$. Now, if $PT$ has more than one leaf, let $p$ be a branch of $PT$ (i.e. a path from a vertex $v_0$ to a leaf $v$, where $v_0$ has more than one child). Then $PT \backslash p$ has one fewer leaf than $PT$, and by inductive hypothesis we have 
    \[
        agr(\PikC(PT \backslash p), w(PT \backslash p) > \epsilon ( |PT| - |p| ). 
    \]
    Furthermore, by Property A, we have that $agr(\PikC(p), w(p)) > \epsilon |p|$. Therefore, 
    \[
        agr(\PikC(PT), w(PT)) = agr(\PikC(PT \backslash p), w(PT \backslash p) + agr(\PikC(p), w(p)) > \epsilon |PT|.
    \]
\end{proof}
\subsection{Sensitive Layered Codes}

\begin{definition}[Sensitive Layered Code] \label{def:listGcode}
    Let $G$ be a layered graph over $\Sigma_{in}$ of depth $n$. A \emph{$\epsilon$-sensitive layered code} for $G$ and alphabet $\Sigma_{out}$ is a $(G, \Sigma_{out})$-code such that for all $w \in \Sigma_{out}^n$ and all $PT \in \PT(\PikC, w, \epsilon)$,
    \[
        agr(\PikC(PT), w(PT)) \le (1 + \epsilon) n. \numberthis \label{eqn:agr-condition}
    \]
\end{definition}

\begin{theorem} \label{thm:listGcode-exist}
    For $\epsilon \in (0, \frac12)$ and a layered graph $G$ over $\Sigma_{in}$ with depth $n \ge \frac{2}{1-\epsilon}$, let $|\Sigma_{out}| > 2 |\Sigma_{in}|)^{6/\epsilon^2}$. Then, a random $(G, \Sigma_{out})$-code is a $\epsilon$-sensitive layered code on $G$ with alphabet $\Sigma_{out}$ with probability at least $1 - 2^{-n/4\epsilon}$.
\end{theorem}

The proof of Theorem~\ref{thm:listGcode-exist} essentially follows from the proof of Theorem 22 in~\cite{BravermanE14}. To prove it, we will need the following two lemmas:

\begin{lemma} \label{lemma:numbersubtrees}
    If $G$ is a layered graph over $\Sigma_{in}$, there exist at most $(|\Sigma_{in}|+1)^{2s}$ rooted subtrees of $G$ of size $s$.
\end{lemma}

\begin{proof}
    Consider the path obtained by conducting a DFS on a rooted subtree, where each symbol indicates which child to go to, and $|\Sigma_{in}|+1$ indicates to go back up the edge traversed downwards to get to the current vertex (note that this edge is unique since we only traverse a subtree). Then, each edge in the subtree is traversed twice. Thus, the number of rooted subtrees of $G$ is at most $(|\Sigma_{in}| + 1)^{2s}$.
\end{proof}

\begin{lemma} \label{lemma:prob-agr-on-tree}
    For any $w \in \Sigma_{out}^n$ and for any collection $PT$ of $s$ edges of $G$, it holds that 
    \[
        Pr [ agr(\PikC(PT), w(PT)) \ge \epsilon s ] \le |\Sigma_{out}|^{-\epsilon s} \binom{s}{\epsilon s} \le |\Sigma_{out}|^{-\epsilon s} 2^s,
    \]
    where randomness is taken over the random choice of layered code $\PikC$ on $G$ with $\Sigma_{out})$.
\end{lemma}

\begin{proof}
    The first inequality follows from the union bound over all possible locations where $\PikC(PT)$ and $w(PT)$ agree, and the second inequality follows from $\binom{s}{\epsilon s} \le 2^s$.
\end{proof}

\begin{proof}[Proof of Theorem~\ref{thm:listGcode-exist}]
    If $w \in \Sigma_{out}^n$ violates~\eqref{eqn:agr-condition}, then there is a prefix tree $PT$ of a subset $S \subseteq L(\PikC, w, \epsilon)$ such that $agr(\PikC(PT), w(PT)) > \max \{ \epsilon |PT|, (1 + \epsilon) n \}$, where $agr(\PikC(PT), w(PT)) > \epsilon |PT|$ is given by Lemma~\ref{lemma:agr>epPT}. To show that such $w$ does not exist, we will show that with high probability over the choice of a random $(G, \Sigma_{out})$-code, $agr(\PikC(PT), w(PT)) \le \max \{ \epsilon |PT|, (1 + \epsilon) n \}$ for all rooted subtrees $PT$ and $w \in \Sigma_{out}^n$. It is enough to prove this claim for all $|PT| \ge (1 + \frac1\epsilon) n$, since if $|PT| < (1 + \frac1\epsilon) n$, then we can extend $PT$ to a tree $PT'$ of size $(1 + \frac1\epsilon) n$ and for this subtree it will hold that $agr(\PikC(PT'), w(PT')) \le (1 + \epsilon) n$ and thus $agr(\PikC(PT), w(PT)) \le (1 + \epsilon) n$. We thus seek to show that with high probability over the choice of a random layered code, $agr(\PikC(PT), w(PT)) \le \epsilon |PT|$ for all rooted subtrees $PT$ of size $\ge (1 + \frac1\epsilon) n$ and $w \in \Sigma_{out}^n$.
    
    Using Lemmas~\ref{lemma:numbersubtrees} and~\ref{lemma:prob-agr-on-tree}, we union bound over all possible trees of size $\ge (1 + \frac1\epsilon) n =: s$ and words $w$ to see that the probability there exists $|PT| \ge (1 + \frac1\epsilon) n$, $w \in \Sigma_{out}^n$ for which $agr(\PikC(PT), w(PT)) \ge \epsilon s$ is upper bounded by
    \begin{align*}
        \sum_{s = (1 + \frac1\epsilon) n}^\infty |\Sigma_{out}|^{-\epsilon s} 2^s \cdot (|\Sigma_{in}| + 1)^{2s} \cdot |\Sigma_{out}|^n
        &= |\Sigma_{out}|^n \sum_{s = (1 + \frac1\epsilon) n}^\infty \left( \frac{2 \cdot (|\Sigma_{in}| + 1)^2}{|\Sigma_{out}|^\epsilon} \right)^s \\
        &\le |\Sigma_{out}|^n \sum_{s = (1 + \frac1\epsilon) n}^\infty \left( \frac{8 \cdot |\Sigma_{in}|^2}{|\Sigma_{out}|^\epsilon} \right)^s
    \end{align*}
    Since $|\Sigma_{out}| > (2 |\Sigma_{in}|)^{6/\epsilon^2} > 8 |\Sigma_{in}|^2$, this is upper bounded by 
    \begin{align*}
        \le |\Sigma_{out}|^n \left( \frac{8 \cdot |\Sigma_{in}|^2}{|\Sigma_{out}|^\epsilon} \right)^{(1 + \frac1\epsilon) n - 1}
        &= \frac{(8 \cdot |\Sigma_{in}|^2)^{(1 + \frac1\epsilon) n - 1}}{|\Sigma_{out}|^{\epsilon n - \epsilon}} \\
        &\le \frac{(8 \cdot |\Sigma_{in}|^2)^{(1 + \frac1\epsilon) n - 1}}{(2\cdot |\Sigma_{in}|)^{6(n-1)/\epsilon}} \\
        &\le \frac{(8 \cdot |\Sigma_{in}|^2)^{(1 + \frac1\epsilon) n - 1}}{(8 \cdot |\Sigma_{in}|^2)^{2(n-1)/\epsilon}} \\
        &\le \left( 8 \cdot |\Sigma_{in}|^2 \right)^{-((1-\epsilon) n - 2)/\epsilon} \\
        &\le 2^{-n/4\epsilon},
    \end{align*}
    where in the last line we use that $\epsilon < \frac12$ and $(1-\epsilon) n \ge 2$.
\end{proof}
\subsection{Decoding}

Sensitive $(G, \Sigma_{out})$ codes will be useful for us because they guarantee that for most locations $i$ on which $\PikC(x)$ and $w$ agree, $w[1:i]$ decodes to $v(x[1:i])$. First, we define decoding.

\begin{definition}[$\PikCDec$]
    Given an $\epsilon$-sensitive-$(G, \Sigma_{out})$-code $\PikC$, we define $\PikCDec$ to be the algorithm that takes as input a string $w \in \Sigma_{out}^i$ and outputs $v \in G$ such that there exists a path $p \in \Sigma_{in}^{i}$ satisfying $\Delta(\PikC(p), w) < 1 - \epsilon$ if exactly one such $v$ exists, and $\perp$ otherwise.
\end{definition}

The main theorem of this section is the following:

\begin{theorem}\label{thm:listGcode-decoding}
    For every $\epsilon, n$, for any layered graph over $\Sigma_{in}$ of depth $n$ and any $\epsilon$-sensitive-$(G, \Sigma_{out})$-code $\PikC : \Sigma_{in}^n \to  \Sigma_{out}^n$, and for any $x \in \Sigma_{in}^n$ and $w \in \Sigma_{out}^n$, let $J$ be the set of indices where $\PikC(x)[i] = w[i]$. For all but at most $2\epsilon n$ values of $i \in J$, it holds that $\PikCDec(w[1:i]) = v(x[1:i])$.
\end{theorem} 

We defer the proof of Theorem~\ref{thm:listGcode-decoding} to after we state a few lemmas.

\begin{lemma} \label{lemma:Li<1}
    Given an $\epsilon$-sensitive-$(G, \Sigma_{out})$-code $\PikC$, for any $w \in \Sigma_{out}^n$ and $\epsilon > 0$, it holds that $|L_i(\PikC, w, \epsilon)| \le 1$ for at least $(1 - \epsilon) n$ values of $i \le n$.
\end{lemma}

\begin{proof}
    Given $w$, we construct $w'$ as follows. Pick a prefix tree $PT$ of $L(\PikC, w, \epsilon)$. For every $i \le n$, define $PT_i(w)$ to be the set of edges in the $i$'th layer of $PT$. If for all $e \in PT_i(w)$ we have that $\PikC(e) \not= w[i]$, then set $w'[i]$ to be $\PikC(e)$ for some arbitrary $e \in PT_i(w)$. Otherwise, set $w'[i] = w[i]$.
    
    Notice that $L(\PikC, w, \epsilon) \subseteq L(\PikC, w', \epsilon)$, since the only indices of $w$ that were changed were those that did not agree with any of the labels of $PT$ in the corresponding layer, so for any path $p(v) \subseteq PT, v \in L_i(\PikC, w, \epsilon)$, it holds that $\Delta_{sfx}(\PikC(p(v)), w'[1:|p(v)|]) \le \Delta_{sfx}(\PikC(p(v)), w[1:|p(v)|]) < 1 - \epsilon$. This means that $PT \in \PT(\PikC, w', \epsilon)$. But by the definition of an $\epsilon$-sensitive-$(G, \Sigma_{out})$-code (Definition~\ref{def:listGcode}), 
    \[
        agr(\PikC(PT), w'(PT))\leq (1 + \epsilon) n.
    \]
    On the other hand, we constructed $w'$ so that in each layer $i$, there is at least one edge on which $\PikC$ and $w'$ agree. Therefore, the number of layers in which there is more than $1$ edge on which $\PikC$ and $w'$ agree is $\le \epsilon n$. In other words, the number of layers in which there is at most $1$ edge on which $\PikC$ and $w'$ agree is at least $(1-\epsilon) n$. Let this set of layers be $I \subseteq [n]$.
    
    Finally, note that for any vertex $v \in L_i(\PikC, w, \epsilon)$ and associated path $p(v) \subseteq PT$, it must hold that $\PikC(p(v))[i] = w[i] = w'[i]$ (otherwise the suffix distance of $\PikC(p(v))$ to $w$ is $1$), so for each of the $\ge (1-\epsilon) n$ layers in $I$, there is at most $1$ vertex $v \in L_i(\PikC, w, \epsilon)$.
\end{proof}

\begin{lemma} [\cite{Gelles-survey}]\label{lem:dec-freq}
    For any $r, s \in \Sigma^n$, if $\Delta(r, s) = \beta n$, then there exists a set of indices $I \subseteq [n]$ of size $|I| \ge (1 - \beta/\alpha)n$ such that for any $i \in I$,
    \[
        \Delta_{sfx}(r[1:i], s[1:i]) < \alpha.
    \]
\end{lemma}

\begin{proof}[Proof of Theorem~\ref{thm:listGcode-decoding}]
    By Lemma~\ref{lem:dec-freq}, there exists a set of indices $I \subseteq [n]$ of size $|I| \ge (1 - \frac{1 - |J|/n}{1 - \epsilon}) n = \frac{|J| - \epsilon n}{1 - \epsilon} \ge |J| - \epsilon n$ such that for any $i \in I$, $\Delta_{sfx}(\PikC(x)[1 : i], w[1 : i]) < 1 - \epsilon$. Note also that $I \subseteq J$, since if $\PikC(x)[i] \not= w[i]$, then $\Delta_{sfx}(\PikC(x)[1:i], w[1:i]) = 1$. 
    
    Furthermore, by Lemma~\ref{lemma:Li<1}, it holds that $|L_i(\PikC, w, \epsilon)| > 1$ on at most $\epsilon n$ values. Thus, there are at least $|J| - 2\epsilon n$ values of $J$ for which $\PikCDec(w[1:i]) = v(x[1:i])$.
\end{proof}


\begin{remark}
    In this section, we defined sensitive layered codes on finite-depth layered graphs. However, our proofs extend straightforwardly to give sensitive layered codes on layered graphs of \emph{infinite depth}. For an infinite graph, sensitivity means that the restriction of the code to any depth $n$ (above a certain threshold) should be a sensitive layered code. It is straightforward via a union bound to see that a random layered code on an infinite layered graph will, with positive probability, satisfy sensitivity.
\end{remark}
\subsection{Discussion} \label{sec:gcode-discussion}

In this section, we have only defined and proven properties of layered codes that are useful in our protocol. However, layered codes also serve as a generalization of tree codes that may be of independent interest, and we hope to see future work further generalizing the results of tree codes to this context. We propose a few problems to guide the future study of layered codes.

\begin{enumerate}
    \item\label{ques:tree-codes} We have shown that \emph{sensitive} layered codes exist, but have not addressed the analogue of tree codes. Do layered codes exist on any layered graph over $\Sigma$? Specifically, for any $\epsilon$ is there an assignment of the edges of a layered graph over $\Sigma$ to a larger alphabet $\Sigma_{out}$ such that for any two words $x,y\in \Sigma^n$ such that $v(x)\neq v(y)$, the suffix distance $\Delta_{sfx}(x,y)>1-\epsilon$?
    \item Our protocol is one in which \emph{layered} codes are necessary, and \emph{tree} codes are not strong enough. Are there other contexts where this is the case? One possible use case may be in low memory settings, where a party cannot remember the full history of the messages they have sent, and so needing only to remember the vertex of the graph they are on may be useful.
    \item Do tree codes beyond layered graphs? For example, does the definition of suffix distance generalize to any directed graph? Does Theorem~\ref{thm:listGcode-exist} generalize to a more general context? Does Question~\ref{ques:tree-codes} generalize?
\end{enumerate}

\section{Positive Rate Scheme Resilient to $\frac16$ Errors} \label{sec:1/6}
In this section, we will formally describe our algorithm to convert any noiseless interactive protocol between Alice and Bob to one that is resilient to $\frac16-\epsilon$ bit flips for any sufficiently small $\epsilon > 0$ (say, $\epsilon < 0.01$), with constant multiplicative blowup in communication complexity and $\tilde{O}(|\pi_0|)$ computational complexity. We note that an error resilience of $\frac16$ is known to be optimal (see Theorem~\ref{thm:EGH16-1/6}). We focus mainly on describing a computationally inefficient scheme, but a recursive application of Corollary~\ref{thm:boosting} results in a computationally efficient scheme.




Throughout this section, let be $\pi_0$ the noiseless protocol of length $n_0$ that Alice and Bob are trying to simulate. Alice's and Bob's private inputs respectively are $x,y\in \{0,1\}^{n_{in}}$ for some $n_{in} \in \bbN$. We assume that $\pi_0$ is alternating (meaning that Alice speaks in the odd rounds and Bob speaks in the even: any protocol can be made alternating with at most a factor of 2 blowup in communication). We also assume that Alice's first message is a $1$. The correct noiseless transcript for $\pi_0$ is denoted $\cT = \cT(x,y)$. We also define $f_x:\{0,1\}^s\to \{0,1\}$ to be the function taking a partial transcript with Bob as the last speaker (only defined on even $s$) and outputs Alice's next message if she has input $x$, as defined by the protocol $\pi_0$. Similarly, we define $f_y : \{ 0, 1 \}^s \to \{0,1\}$ to be the function taking a partial transcript with Alice as the last speaker and outputs Bob's next message on input $y$ as defined by $\pi_0$. We say a transcript $T$ is \emph{inconsistent with $x$} if for some even $s$ with $|s|<|T|$, if $f_x(T[1:s])\neq T[s+1]$, and similarly \emph{inconsistent with $y$} if for some odd $s$, $f_y(T[1:s])\neq T[s+1]$. 

We denote a parameter $\epsilon>0$, where the adversary will be permitted to flip $\frac16-O(\epsilon)$ bits.

\subsection{Preliminaries and Definitions} \label{sec:1/6-prelims}

In our protocol, Alice and Bob will each track a guess for the noiseless transcript $\cT$. Specifically, they will track a sequence of updates denoted $U_A, U_B\in \{0,1,\rewind,\bullet\}^*$ that evaluates to their current guess for $\cT$. Generally, Alice's guess is odd length (meaning $|t(v(U_A))|$ is odd) since she speaks on odd turns in $\pi_0$, and Bob's guess $t(v(U_B))$ is even length. The exception is if Alice has a transcript that is either length $0$ or length $n_0$. Roughly, an update of $0$ or $1$ adds this bit onto the transcript, an update of $\rewind$ rewinds the previous bit of the transcript, and an update of $\bullet$ keeps the transcript the same. After each message, the receiving party will append some new updates to this sequence based on the other person's message. We begin with some necessary definitions.

\subsubsection{Transcript Graph} \label{sec:transcript_graph}
We begin by informally describing the layered graph that the parties use to build their transcript guesses. The vertices of $G$ at a given layer $\ell$ describe the possible transcript guesses for the noiseless protocol that a party could have after appending $\ell$ edges $\in \{0,1,\rewind,\bullet\}^*$ as updates to the transcript guess. The depth of the graph is $K=\frac{n_0}{\epsilon}$.  

\begin{definition} [Transcript Graph ($G$)]\label{def:G} 
Let $G$ be the following particular instance of a layered graph over the alphabet $\{ 0, 1, \rewind, \bullet \}$ (see Definition~\ref{def:layered-graph}).
\begin{itemize}
    \item At every layer $\ell \in [0, K]$, the vertices are all  elements of the form $\{0,1\}^{\leq \ell}_\ell$ (for example, at layer $5$, a possible vertex is $01_5$).
    For a vertex $v$ denoted $v=y_\ell$, where $y\in \{0,1\}^*$ and $\ell\in \bbN$, define $t(v):=y\in\{0,1\}^*$ and $\ell(v):=\ell$. The set of all vertices of $G$ is denoted $\Pi$. 
    \item The out-edges from a given node $v$ in some layer $< K$ are $0,1,\rewind,\bullet$. For an edge $e \in \{ 0, 1, \rewind, \bullet \}$, the node $v \oplus e$ at the end of the out-edge from $v$ labeled $e$ is computed as follows
    \[
    v\oplus e := \begin{cases}
        (t(v)||e)_{\ell(v)+1} & e\in\{0,1\}\\
        (t(v)[1:|t(v)|-1])_{\ell(v)+1}& e=\rewind \text{ and } y\neq\emptyset\\
        \emptyset_{\ell(v)+1} & e=\rewind \text{ and } t(v)=\emptyset\\
        t(v)_{\ell(v)+1} & e=\bullet
    \end{cases}.
    \]
    Vertices in layer $K$ have no out-edges.
\end{itemize}
\end{definition}


As shorthand, for a layered code $\PikC$ on $G$, and for $v\in \Pi$ and $p\in \Sigma^*$, let $\PikC(v,p)\in \Sigma^{|p|} := \PikC(H)$ where $H$ is the subgraph of $G$ corresponding to the path starting at $v$ obtained by following the edges specified by $p$.

\subsubsection{Transcript Operations and Instructions}

Along with $U_A$ and $U_B$, Alice and Bob track a weight (confidence) $w_A$ and $w_B$ associated with this guess. We will have that $w=0$ unless $T$ is a complete transcript. A message received from the other party will contain an \emph{instruction} for how to update $(U,w)$. The instruction is in $\{0,1,\rewind,\bullet\}$. 

We define some functions that describe the updates that Alice and Bob make to $(U_A, w_A)$ and $(U_B, w_B)$. We begin with the definition of $\op_x(T)$ and $\op_y(T)$. This function takes a partial transcript $T\in \{0,1\}^*$\footnote{Notice that $T\in \{0,1\}^*$ while each party tracks $U\in \{0,1,\rewind,\bullet\}^*$. Each $U$ evaluates to a transcript $t(v(U))\in \{0,1\}^*$ which corresponds to the input to $\op$.} and calculates the instruction that the party with $x$ or $y$ gives to extend $T$. The function is defined on every possible partial transcript $T$, but only takes on a meaningful value when the party with the corresponding $x$ or $y$ is the next to speak, or if the transcript is complete (of length $n_0$).

\begin{definition}[$\op_{r}(T)$]
    We define $\op_r(T) : \{0,1\}^{\leq n_0} \to \{0,1,\rewind\}$, for $r\in \{x,y\}$. 
    Let the set $S$ denote the set of lengths of $T$ on which $f_r$ is defined:
    $S$ is all the even indices $<n_0$ if $r=x$ or all the odd indices $<n_0$ if $r=y$.
    
    \begin{itemize}
        \item If $T$ is inconsistent with $r$, then $\op_r(T) = \rewind$. 
        \item Else if $|T|\in S$, then $\op_r(T)  = f_r(T)$.
        \item Else, $\op_r(T) = 1$. 
    \end{itemize}
\end{definition}

The final condition which results in a ``default'' response of $\op_r(T) = 1$ occurs in one of two cases: when the party with input $r$ is not the next to speak, allowing $1$ to serve as a meaningless instruction, or when the transcript is complete (of length $n_0$) and the party wants to indicate it is consistent with their input.

Next, we define the function $\op_{T'}(T)$, where $T'$ is a complete transcript. The function $\op_{T'}(T)$ takes a partial transcript $T$ and returns the instruction that brings it one step closer to $T'$.

\begin{definition}[$\op_{T'}(T)$]
    Let $T' \in \{0,1\}^{\leq n_0}$ with $|T'|=n_0$. We define $\op_{T'}(T) : \{0,1\}^{\leq n_0} \to \{0,1,\rewind\}$ as follows.
    
    \begin{itemize}
        \item If $T'=T$, then $\op_{T'}(T) = 1$. 
        \item Else, if $T$ is a strict prefix of $T'$, then $\op_{T'}(T) = T'[|T|+1]$. 
        \item Else, $\op_{T'}(T) = \rewind$.
    \end{itemize}
\end{definition}


Next, we define a function that Alice and Bob use to update their transcript guess $U_A$ or $U_B$ and weight $w_A$ or $w_B$ when they receive an instruction. 
Every time a party receives a message, the party adds two edges onto their guess $U_A$ or $U_B$: namely the update $\hat{\delta}\in \{0,1,\rewind,\bullet\}$ that they deduce from the other party's message, and their own response to that addition.\footnote{They will also add two more edges, corresponding to $\bullet \bullet$, to account for parity issues, but we leave this discussion for later. We also do not yet discuss how they deduce $\hat{\delta}$ from the other party's message.} 
Again, recall that Alice's partial transcript guess $t(v(U_A))$ is of odd or exactly $0$ or $n_0$ length, and Bob's guess $t(v(U_B))$ is of even length. 

\begin{definition}[$(U,w)\otimes_r \hat{\delta}$] \label{def:otimes}
Let $r\in \{x,y\}$. Given a sequence of updates $U\in \{0,1,\rewind,\bullet\}^*$, an instruction $\hat{\delta}\in \{0,1,\rewind,\bullet\}$, and weight $w\in \bbN$, return a new pair $(U',w') \gets (U,w)\otimes_r \hat{\delta}$ as follows. As before, let the set $S$ denote the set of lengths of $T \in \{ 0, 1 \}^*$ on which $f_r$ is defined: $S$ is all the even indices $<n_0$ if $r=x$ and all the odd indices $< n_0$ if $r=y$. 

\begin{itemize}
\item If $\hat{\delta}=\bullet$: 

    Let $U'= U||\bullet||\bullet$ and $w'=w$.

\item If $\hat{\delta}=\rewind$:

    If $w>0$, then let $U'= U||\bullet||\bullet$ and $w'=w-1$. 

    Otherwise, if $|t(v(U))|-1\in S$, then let $U'=U||\rewind||\rewind$ and $w' = w$. Else, $|t(v(U))| \in S$, and let $U'=U||\rewind||\bullet$ and $w' = w$. 
    

\item If $\hat{\delta}=0$ or $\hat{\delta}=1$:

    Let $T = t(v(U))$. If $|T| = n_0$, then $U'= U || \bullet || \bullet$ and $w'=w+1$.
    
    Otherwise, if $|T| - 1 \in S$: if $|T|<n_0-1$, then $U'=U||\hat{\delta}||\op_{r}(t(v(U||\hat{\delta})))$, and if $|T|=n_0-1$, then $U'=U||\hat{\delta}||\bullet$. Else if $|T|\in S$, then $U'=U||\bullet||\op_{r}(T)$. In any case, $w'=0$.

\end{itemize}
\end{definition}

Notice that in every case, the path $U'$ is an extension of $U$ with two additional letters.

\subsubsection{The Error Correcting Code}

Finally, we define the error correcting code $\ECC$ that Alice and Bob use to encode the letters of the large alphabet layered code. 


\begin{lemma}[\cite{GuptaZ22a}] \label{lemma:ecc}
    There exists an explicit error correcting code
    \[
        \ECC_{\Sigma,\epsilon} := \Sigma^2\times \{0,1,\rewind,?\} \rightarrow \{ 0, 1 \}^{M(|\Sigma|, \epsilon)}
    \]
    for some $M(|\Sigma|, \epsilon) = O_\epsilon(|\Sigma|)$ with the following properties:
    \begin{itemize}
        \item For any $z_0 \not= z_1 \in \Sigma^2$ and $\delta_0, \delta_1 \in \{ 0, 1, \rewind, ? \}$,
        \[
            \Delta \big( \ECC_{\Sigma,\epsilon}(z_0, \delta_0), \ECC_{\Sigma,\epsilon}(z_1, \delta_1) \big) \ge \left( \frac12 - \epsilon \right) \cdot M(|\Sigma|, \epsilon), \numberthis \label{eqn:ecc-1/2}
        \]
        \item For any $z \in \Sigma^2$ and $\delta_0 \not= \delta_1 \in \{ 0, 1, \rewind, ? \}$, 
        \[
            \Delta \big( \ECC_{\Sigma,\epsilon}(z, \delta_0), \ECC_{\Sigma,\epsilon}(z, \delta_1) \big) \geq \frac23M(|\Sigma|, \epsilon). \numberthis \label{eqn:ecc-2/3}
        \]
    \end{itemize}
\end{lemma}

\noindent
We remark that due to the distance conditions, for any fixed $z'$ and any string $s \in \{ 0, 1 \}^{M(|\Sigma|, \epsilon)}$, at most one of the following holds:
\begin{itemize}
    \item There exists $\delta \in \{ 0, 1, \rewind, ? \}$ such that $\Delta(s, \ECC_{\Sigma,\epsilon}(z', \delta)) < \frac13$.
    \item There exists $z \in \Sigma^2, \delta \in \{ 0, 1, \rewind, ? \}$ such that $\Delta(s, \ECC_{\Sigma,\epsilon}(z, \delta)) < \frac16 - \epsilon$.
\end{itemize}
In particular, the three cases in Protocol~\ref{prot:1/6} are disjoint.
\subsection{The Inefficient, Positive Rate Protocol} \label{sec:1/6-protocol}

We are now ready to state our (inefficient) positive rate protocol that is resilient to $\frac16 - \epsilon$ errors.

Recall that $\pi_0$ is an alternating protocol of length $n_0$, such that Alice speaks first and her first message is always a $1$. Let $\PikC$ be a $\epsilon$-sensitive-$(G,\Sigma)$-code for some alphabet $\Sigma$ of size $O_\epsilon(1)$. Note that Alice and Bob can agree on an explicit choice of $\PikC$, for example by both choosing the lexicographically first such code (it takes up to $2^{2^K}$-time to find such a code). Also let $\ECC = \ECC_{\Sigma, \epsilon}$ be the error correcting code from Lemma~\ref{lemma:ecc}.

Before we state our protocol formally in Section~\ref{sec:formal-protocol}, we give an explanation of the protocol. While Section~\ref{sec:overview-pi^2} and Section~\ref{sec:overview-graph-codes} give an explanation of the ideas in our protocol, this section explains how we implement them. In this explanation, we first focus on when Eve corrupts a message either entirely to another valid message, or not at all. We talk about the protocol from Alice's perspective (Bob is symmetric).

Recall that Alice tracks a guess for the sequence of updates $U_A\in \{0,1,\rewind,\bullet\}^*$ along with a confidence weight $w_A \ge 0$. The sequence of updates in $U_A$ describes Alice's guess for the transcript: her transcript guess $\in \{ 0, 1 \}^{\le n_0}$ is simply the result of applying the updates to the empty string. 

Every round, Alice sends one of two things: she either asks her own question (a message of the form $\ECC(z,?)$, where $z$ lets Bob deduce $U_A$ which specifies her transcript guess), or she sends an answer to Bob's question (a message of the form $\ECC(z,\delta\in\{0,1,\rewind\})$ where $z$ reflects the transcript she believes Bob has asked about). Likewise, Bob always sends a question $\ECC(z,?)$ or an answer $\ECC(z, \delta \in \{ 0, 1, \rewind \})$. We will discuss later what $z$ should look like.

Whenever Alice receives a message $\ECC(z_B, \delta \in \{ 0,1,\rewind,? \})$ from Bob, she updates $w_A$ and $U_A$ based on the received message and history. She then chooses to send either a question or an answer. Specifically:
\begin{itemize}
\item 
    If Alice receives an answer $\ECC(z_B, \delta \in \{ 0, 1, \rewind \})$ where $z_B$ matches her own transcript guess, she updates $(U_A, w_A)$ accordingly by setting $(U_A, w_A) \gets (U_A, w_A) \otimes_x \delta$. This consists of (with probability $1$) appending two symbols to $U_A$ and possibly adjusting the weight $w_A$ so that she has overall updated in the direction specified by $\delta$. She then asks a question.
\item 
    If she instead receives a question $\ECC(z_B, ?)$, she uses $z_B$ and the history of received messages to make a guess for the full sequence of updates $U^*_B$ that Bob has made. $T^*_B = t(v(U^*_B))$ is then her understanding of Bob's current transcript guess. 
    \begin{itemize}
    \item 
        If $T^*_B$ is a partial transcript or is inconsistent with $x$, she updates $(U_A, w_A) \gets (U_A, W_A) \otimes_x \bullet$ (``do nothing''). She then sends an answer $\ECC(z_A, \delta = \op_x(T^*_B) \in \{ 0, 1, \rewind \})$. 
    \item 
        Else if $T^*_B$ is a complete transcript (length $n_0$) that is also consistent with $x$, she updates $U_A$ with probability $0.5$ in the direction of $T^*_B$, i.e. by computing $(U_A, w_A) \gets (U_A, w_A) \otimes_x \op_{T^*_B}(t(v(U_A)))$. This consists of appending two symbols to $U_A$ and possibly adjusting $w_A$. She then asks a question.
    \end{itemize}
    In the special case that $t(v(U_B)) =: T_B = T_A := t(v(U_A))$, i.e. Bob's current transcript guess is the same as Alice's (because Alice and Bob's transcripts are usually different parity lengths, this can only happen if $T_B = T_A$ are both the same complete transcript or both the empty transcript), Alice asks a question. Bob will interpret her question $\ECC(z_A, ?)$ as both an answer of $1$ (extending his complete transcript guess or empty transcript) \emph{and} a question. That is, if Bob receives Alice's message correctly, he will both update $(U_B, w_B)$ (with probability $1$) via the operation $\hat{\delta} = 1$ \emph{and} send his question. Note that in both the case $T_B = T_A = \cT$ or $T_B = T_A = \emptyset$ the update $\hat{\delta} = 1$ causes a good update, since we assumed Alice's first message is always a $1$.
\end{itemize}

We emphasize that every time Alice updates (after receiving a message from Bob), she appends \emph{two} elements $\in \{ 0, 1, \rewind, \bullet \}$ to $U_A$, so that the resulting transcript guess $t(v(U_A))$ still ends on her speaking. (The exception is when $t(v(U_A))$ is a complete transcript of length $n_0$ or the empty transcript of length $0$: then, Alice still appends two update instructions, but the resulting transcript may be of even ($n_0$ or $0$) length.)


\paragraph{The token $z$.}

When Alice is asking a question $\ECC(z, ?)$, we need $z$ to allow Bob to determine Alice's current transcript guess $T_A = t(v(U_A))$. Note that sending $z = U_A$ (or even $z = T_A$) is too long. Instead, Alice simply sends $z \in \Sigma^2$ to be her most recent updates to $U_A$, i.e. the last two operations she appended to $U_A$, encoded into a tree code. Then many of Alice's messages (the ones where she asked a question) are symbols of the tree code encoding of $U_A$, which will be sufficient for Bob to determine $U_A$.

In the case where Alice answers Bob's question, her message is of the form $\ECC(z, \delta \in \{ 0, 1, \rewind \})$, where $z$ must, in some way, echo Bob's question so that Bob can tell that she is answering the right question. As before, she cannot send $z$ as the entire belief of Bob's transcript guess $t(v_B)$ where $v_b \in \Pi$ is a vertex of $G$, because this is too long. Instead, $z$ will be $\in \Sigma^2$ and will be dependent on her current belief about Bob's current transcript guess (as a vertex $v_B$ in the transcript graph $G$). It is almost okay to let $z$ be exactly $z'$, if she just received $\ECC(z', ?)$ from Bob so that $z' \in \Sigma^2$ are the last two tree code symbols in the encoding of $U_B$; however this causes a misalignment in $\ell(v_B)$ and the length of $U_A$ that requires a different convention to fix.


To elaborate, when Alice asks a question, she sends the last two symbols of the tree code at indices $|U_A|-1$ and $|U_A|$. When she answers Bob's question, she might want to send the symbols at positions $|U_B|$ and $|U_B|-1$ of what she believes to be Bob's update sequence $U_B$. However, $U_B$ (which has length $\ell(v_B)$) is shorter than $U_A$, since it was last updated on the previous message. 
This clashes with our requirement that when Alice and Bob both have the correct transcript $\cT$ as the evaluation of their guesses $U_A$ and $U_B$, then Bob must interpret the token $z$ in Alice's message as the same regardless of whether she is asking or answering a question.
To resolve this, we say that after she decodes Bob's message to $v_B$, she adds $\bullet\bullet$ onto it; this makes it the same length as $U_A$, and then she responds with the last two symbols of the new encoding $\PikC(v_B, \bullet\bullet)$. Additionally, every time she updates $U_A$, she first updates $U_A$ with $\bullet\bullet$ (as a space holder that says ``do nothing''). The result is that both $U_A$ and $U_B$ increase in length by $4$ every time the corresponding party receives a message and makes an update. For instance, after Bob has sent the $k$'th message (so both Alice and Bob have sent $k/2$ messages), Alice updates so that $U_A$ goes from length $2(k-1)$ to length $2(k+1)$, where the first two updates are simply $\bullet\bullet$ and the next two correspond to the additions to $U_A$. Meanwhile, $U_B$ is of length $2k$, so if she wishes to answer $v_B = v(U_B)$, she would add $\bullet\bullet$ to $v_B$ to make it length $2(k+1)$ as well, and then send the last two symbols in the tree code encoding.

Finally, we discuss a point glossed over so far: how Alice actually decodes Bob's question to $v_B$ if she only receives the encoding of the most recent two symbols $z\in \Sigma^2$ of his transcript guess $U_B$. She tracks $P_A \in (\Sigma^2)^*$ as a history of all the symbols $\in \Sigma^2$ that she and Bob have sent. That is, every time she sends or receives a message $\ECC(z \in \Sigma^2, \delta)$, she appends $z$ to $P_A$. Note that $P_A$ has the correct symbols of the tree code encoding of $U_B$ whenever Alice correctly receives Bob's question. Theorem~\ref{thm:listGcode-decoding} says that most of the time when Alice correctly receives Bob's question $\ECC(z, ?)$, she can decode his entire tree code encoding of $U_B$ correctly (even though many elements of $P_A$ do not even correspond to Bob's messages!). 

To remember the rules for $U_A$ and $P_A$, it is helpful to keep in mind the following picture. After Alice speaks in the $k$'th round, i.e. a total of $k$ messages by either Alice or Bob have been sent so far, both $U_A$ and $P_A$ should be of length $2k$. $U_A$ is of the form $\dots || \bullet\bullet || (\delta_B\delta_A)_{k-2} || \bullet\bullet || (\delta_B\delta_A)_{k}$. That is, entries of $U_A$ that are $\bullet\bullet$ are when Bob is talking. Meanwhile, $P_A$ is of the form $\dots || z_{B,k-3} || z_{A,k-2} || z_{B,k-1} || z_{A,k}$, where $z_{A,i}$ corresponds to the symbols she sent in round $i$, and $z_{B,i}$ corresponds to the symbols she received in round $i$. 

\paragraph{Partial Corruptions.}

Lastly, we mention how we handle partial corruptions, i.e. if a received message is not a codeword. The receiver will choose a nearby codeword (with distance $< \frac13$ if the codeword is an answer to the party's last question, or with distance $\frac16 - \epsilon$ if the codeword is a question). With probability proportional to the distance from the codeword, they default to sending a question. Otherwise, they will respond to that codeword as we have described above. 

\paragraph{Summary.}

A brief summary of the most important details:
\begin{itemize}
    \item Every message Alice sends is of the form $ECC(z\in \Sigma^2,\delta\in\{0,1,\rewind,?\})$. The instruction $\delta$ is $?$ if Alice is asking Bob a question (potentially also responding to his question), and $0,1$ or $\rewind$ if she is only responding to his question.
    \item After receiving a message, Alice performs four updates to both $U_A$, appending $\bullet\bullet$ and two symbols in $\{0,1,\rewind,\bullet\}$. She similarly performs four updates to $P_A$, appending the two symbols $z^* \in \Sigma^2$ received in Bob's message and then appending the two symbols $z$ that she is sending in her own next message. 
    \item After sending message $k$, $U_A$ and $P_A$ are both length $2k$.
    \item Partial corruptions are handled by performing the behavior described in this section with probability linearly decreasing with the distance to a nearby codeword. The default message is a question.
\end{itemize}

\paragraph{Indexing: Notational Change.}

Thus far, we have described $U_A$ and $P_A$ as being a length $2k$ sequence of symbols in $\{ 0, 1, \rewind, \bullet \}$ and $\Sigma$ respectively, where Alice has just sent the $k$'th message. Note however that symbols are always appended to $U_A$ and $P_A$ in pairs. Thus, we can instead regard the alphabets of $U_A$ and $P_A$ as being pairs of updates/layered code symbols instead. Throughout the rest of this section, we instead regard $U_A \in ( \{ 0, 1, \rewind, \bullet \}^2 )^*$ and $P_A \in ( \Sigma^2 )^*$, so that after Alice sends the $k$'th message both $U_A$ and $P_A$ are length $k$. Then, for instance $U_A[k]$ denotes the last two updates Alice has made to $U_A$, while $U_A[k-1] = \bullet\bullet$.

Similarly, the alphabet of $\PikC(U_A)$ is $\Sigma^2$, so that $\PikC(U_A)$ is of length $k = |U_A|$. For instance, $\PikC(U_A)[|U_A|]$ are the last two symbols of $\PikC(U_A)$.






\subsubsection{Formal Description of Protocol} \label{sec:formal-protocol}

\protocol{Inefficient, Positive Rate Scheme Resilient to $\approx \frac16$ Errors}{1/6}{
    Recall that $\pi_0$ is a an alternating, noiseless protocol of length $n_0$, such that Alice speaks first and her first message is a $1$. Alice and Bob have inputs $x$ and $y$ respectively, determining their behavior in this protocol. The noiseless protocol has transcript $\cT=\cT(x,y) \in \{ 0, 1 \}^{n_0}$. Our error-resilient protocol consists of $K = \frac{n_0}{\epsilon}$ messages numbered $1, \dots, K$, each consisting of $M(|\Sigma|, \epsilon) = O_\epsilon(1)$ bits. Alice sends the odd messages and Bob sends the even. 
    
    
    Recall that $\PikC$ is an $\epsilon$-sensitive layered code of $G$ with the alphabet $\Sigma$. Alice and Bob first (non-interactively) agree on an explicit choice of $\PikC$ by testing each labeling of $G$ and taking the lexicographically first layered code that is $\epsilon$-sensitive.
    
    \vspace{0.075in}
    
    Alice and Bob track a private sequence of updates of the transcript guess, denoted $U_A,U_B \in \{ 0, 1, \rewind, \bullet \}^2 )^*$ respectively initialized to $\emptyset$.
    They also track confidence weights $w_A,w_B\in \bbN$, both initialized to $0$. Alice and Bob additionally track the sequence $P_A, P_B \in (\Sigma^2)^*$ of pairs of symbols $\in \Sigma^2$ that they have sent and received throughout the protocol. $P_A, P_B$ are both initialized to $\emptyset$.
    
    \vspace{0.075in}
    
    In what follows, we describe Alice's behavior. Bob's behavior is identical, except notationally switching $x$ and $y$, and $A$ and $B$. At the end of the protocol, Alice and Bob output $(t(v(U_A)),\frac{2w_A}{K})$ and $(t(v(U_B)),\frac{2w_A}{K})$ respectively.
    
    Alice's first turn is special; she sets $U_A=\bullet1$, sets $P_A=\PikC(\bullet1)$, and sends $\ECC(\PikC(\bullet1),?)$. 
    
    \begin{center}
    {\bf \fbox{Alice}}
    \end{center}
    
    Alice has just received a message $m$ from Bob. Let $\asked=\true$ 
    if the last message she sent was of the form $\ECC(z,?)$ for some $z\in \Sigma^2$ and $\false$ otherwise (we let $\asked=\false$ in the first round for Bob). Let $d_m(z, \delta)$ denote $\frac1{M(|\Sigma|, \epsilon)} \cdot \Delta(m, \ECC(z, \delta))$.
    
    Alice sets $(U_A,w_A)\gets (U_A,w_A)\otimes_x\bullet$ and $z_A\in \Sigma^2$ to be $\PikC(U_A)[|U_A|]$. 
    Then, she picks the first of the following cases that holds.
    
    \begin{caseofb}
    
    \caseb{$\asked=\true$ and for some $\delta\in\{0,1,\rewind,?\}$, we have $d_m(z_A,\delta) < \frac13$.\label{case1}}{
        Let $p=1-3d_m(z_A,\delta)$. 
        \begin{itemize}
           \item Let the instruction $\hat{\delta}=\delta$ unless $\delta=?$, in which case $\hat{\delta}=1$. Alice sets $(U_A,w_A)\gets (U_A,w_A)\otimes_x\hat{\delta}$ and otherwise (with probability $1-p$), sets $(U_A,w_A)\gets (U_A,w_A)\otimes_x\bullet$.
           She computes $\zeta = \PikC(U_A)[|U_A|]$.
           \item Alice sets $P_A \gets P_A || z_A || \zeta$.
           \item Alice sends $\ECC(\zeta,?)$.
        \end{itemize}
    }
        
    \caseb{For some $z^* \in \Sigma^2$, we have $d_m(z^*,?) \le \frac16 - \epsilon$.\label{case2}} {
        Alice computes $v^*=\PikCDec(P_A||z^*)$.
        \begin{subcaseofb}
        
        \subcaseb{$v^*=\perp$.\label{sub21}}{
            \begin{itemize}
                \item Alice sets $(U_A,w_A)\gets (U_A,w_A)\otimes_x\bullet$. Alice sets $\zeta = \PikC(U_A)[|U_A|]$.
                \item Alice sets $P_A \gets P_A || z^* || \zeta$. 
                \item Alice sends $\ECC(\zeta, ?)$.
            \end{itemize}
        }
            
        \vspace{0.25cm}
        In the next two subcases, $v^* \in \Pi$. Let $T^* = t(v^*)$.
            
        \subcaseb{$T^*$ is complete, i.e. $|T^*| = n_0$, and is consistent with $x$.\label{sub22}}{
            Let $p=0.5-3d_m(z^*,?)$.
            \begin{itemize}
                \item Alice computes $\hat{\delta}=\op_{T^*}(t(v(U_A)))$. With probability $p$, Alice sets $(U_A,w_A)\gets (U_A,w_A)\otimes_x \hat{\delta}$ and otherwise (with probability $1-p$), sets $(U_A,w_A)\gets (U_A,w_A)\otimes_x\bullet$.
                She sets $\zeta = \PikC(U_A)[|U_A|]$.
                \item Alice sets $P_A \gets P_A || z^* || \zeta$.
                \item Alice sends $\ECC(\zeta,?)$.
            \end{itemize}
        }
        
        \subcaseb{$|T^*| \neq n_0$ or $T^*$ is inconsistent with $x$.\label{sub23}}{
            Let $p=1-6d_m(z^*,?)$.
            \begin{itemize}
                \item Alice sets $(U_A,w_A)\gets (U_A,w_A)\otimes_x\bullet$.
                \item With probability $p$, Alice computes $\delta=\op_x(T^*)$ and sends $\ECC(\zeta := \PikC(v^*,\bullet\bullet),\delta)$. 
                
                Else (with probability $1-p$), she sends $\ECC(\zeta := \PikC(U_A)[|U_A|],?)$.
                
                \item Alice sets $P_A \gets P_A || z^* || \zeta$.
            \end{itemize}
        }
        \end{subcaseofb}
    }
    \caseb{None of the above.\label{case3}}{
        \begin{itemize}
            \item Alice sets $(U_A,w_A)\gets (U_A,w_A)\otimes_x\bullet$. She computes $\zeta = \PikC(U_A)[|U_A|]$.
            \item Alice sets $P_A\gets P_A||z||\zeta$, where $z \in \Sigma^2$ is some arbitrary pair of symbols. 
            \item Alice sends $\ECC(\zeta,?)$.
        \end{itemize}
    }
    
    \end{caseofb}
}

\subsection{Main Theorems} \label{sec:1/6-theorem}

\begin{theorem} \label{thm:1/6}
    Protocol~\ref{prot:1/6} is a $\left(\frac16,1224\epsilon,2\cdot \exp \left(-\frac{\epsilon n_0}{800}\right)\right)$-scaling scheme with communication complexity $O_\epsilon(n_0)$ and computational complexity $2^{2^{O_\epsilon(n_0)}}$.
\end{theorem}

We prove Theorem~\ref{thm:1/6} in Section~\ref{sec:1/6-analysis}. Combining Theorem~\ref{thm:1/6} with the boosting procedure in Protocol~\ref{prot:boosting}, we obtain the following result.

\begin{corollary} \label{cor:1/6-efficient}
    For any $\epsilon>0$ there is a scheme for noiseless protocols of length $n_0$ that is resilient to $\left( \frac16 - \epsilon \right)$-fraction of errors with probability $1 - e^{-\epsilon n_0/40 \log^4 n_0}$. The scheme has communication complexity $O_\epsilon(n_0)$ and computational complexity $\tilde{O}_\epsilon(n_0)$.
\end{corollary}

\begin{proof}
    Let $\epsilon' = \epsilon/256$, and let $C_\epsilon$ be such that $e^{-\epsilon' C_\epsilon / 10\log^4 C_\epsilon} < \epsilon'$. We choose $C_\epsilon \ge \frac{8 \cdot 800 \cdot 1224}{\epsilon'^2}$ so that $C_\epsilon \ge \frac{100}{\epsilon'} + 1$ and $\frac{\epsilon'}{4} > 2 \cdot \exp(- \frac8{\epsilon'} \cdot \log^4 n_0) \ge 2 \cdot \exp ( - \frac{\epsilon' C_\epsilon \log^4 n_0}{800 \cdot 1224} )$.
   
    We recursively apply Theorem~\ref{thm:boosting} three times.
    \begin{itemize}
    \item 
        We begin with the $(\frac16, \epsilon', 2 \cdot \exp(-\frac{\epsilon' n_0}{800 \cdot 1224})$-scaling scheme from Theorem~\ref{thm:1/6}, which has communication complexity $O_{\epsilon}(n_0)$ and computational complexity $\exp(\exp_\epsilon(n_0))$.
    \item 
        Since $2 \cdot \exp(-\frac{\epsilon' C_\epsilon \log^4 n_0}{800 \cdot 1224}) < \frac{\epsilon'}4$, we apply Theorem~\ref{thm:boosting} to obtain a $(\frac16, 4\epsilon', e^{-\epsilon' n_0 / 10 \log^4 n_0})$-scaling scheme with communication complexity $\frac{n_0}{\epsilon' \log^4 n_0} \cdot O_{\epsilon'}(C_\epsilon \log^4 n_0) = O_{\epsilon}(n_0)$ and computational complexity $\tilde{O}_{\epsilon'}(n_0) \cdot \exp(\exp_\epsilon(C_\epsilon \log^4 n_0)) = \exp(\exp_{\epsilon}(\polylog n_0))$. Let $\mu'_{\epsilon'}(n_0) = e^{-\epsilon' n_0 / 10 \log^4 n_0}$.
    \item 
        Next, since $\mu'_{\epsilon'}(C_\epsilon \log^4 n_0) = \exp(-\frac{\epsilon' C_\epsilon \log^4 n_0}{10 \log^4 (C_\epsilon \log^4 n_0)}) \le \exp(-\frac{\epsilon' C_\epsilon}{10 \log^4 C_\epsilon}) < \epsilon' = \frac{4\epsilon'}{4}$, we can apply Theorem~\ref{thm:boosting} again to obtain a $(\frac16, 16\epsilon', e^{-2\epsilon' n_0/5 \log^4 n_0})$-scaling scheme with communication complexity $\frac{n_0}{4\epsilon' \log^4 n_0} \cdot O_{\epsilon}(C_\epsilon \log^4 n_0) = O_{\epsilon}(n_0)$ and computational complexity $\tilde{O}_{\epsilon'}(n_0) \cdot \exp(\exp_{\epsilon}(\polylog(C_\epsilon \log^4 n_0))) = \exp(\exp_{\epsilon'}(\poly(\log\log(n_0))))$. Let $\mu''_{\epsilon'}(n_0) = e^{-2\epsilon' n_0 / 5\log^4 n_0}$.
    \item 
        Again, since $\mu''_{\epsilon'}(C_\epsilon \log^4 n_0) = \exp(-\frac{2\epsilon' C_\epsilon \log^4 n_0}{5 \log^4 (C_\epsilon \log^4 n_0)}) \le \exp(- \frac{2\epsilon' C_\epsilon}{5\log^4 C_\epsilon}) < \epsilon'^4 < \frac{16\epsilon'}{4}$, we can apply Theorem~\ref{thm:boosting} to get a $(\frac16, 64\epsilon', e^{-8\epsilon' n_0/5\log^4 n_0})$-scaling scheme with communication complexity $\frac{n_0}{16\epsilon' \log^4 n_0} \cdot O_{\epsilon}(C_\epsilon \log^4 n_0) = O_{\epsilon}(n_0)$ and computational complexity $\tilde{O}_{\epsilon}(n_0) \cdot \exp(\exp_{\epsilon}(\poly(\log\log(C_\epsilon\log^4 n_0)))) = \exp(\exp_\epsilon(\poly(\log\log\log n_0))) \leq \poly_{\epsilon}(n_0)$. Let $\mu'''_{\epsilon'}(n_0) = e^{-8\epsilon' n_0 / 5\log^4 n_0}$.
    \item 
        Finally, to further reduce the computational complexity to $\tilde{O}_{\epsilon}(n_0)$, we apply Theorem~\ref{thm:boosting} one last time. Since $\mu'''_{\epsilon'}(C_\epsilon \log^4 n_0) = \exp(-\frac{8\epsilon' C_\epsilon \log^4 n_0}{5 \log^4 (C_\epsilon \log^4 n_0)}) \le \exp (- \frac{8\epsilon'C_\epsilon}{5 \log^4 C_\epsilon}) < \epsilon'^{16} < \frac{64\epsilon'}{4}$, we get a $(\frac16, 256\epsilon', e^{-32\epsilon' n_0/5\log^4 n_0})$-scaling scheme with communication complexity $\frac{n_0}{64\epsilon' \log^4 n_0} \cdot O_{\epsilon}(C_\epsilon \log^4 n_0) = O_{\epsilon}(n_0)$ and computational complexity $\tilde{O}_{\epsilon}(n_0) \cdot \poly_\epsilon(C_\epsilon\log^4 n_0) = \tilde{O}_{\epsilon}(n_0)$.
    \end{itemize}
    Thus, we have arrived at a $(\frac16, \epsilon, e^{-\epsilon n_0/40 \log^4 n_0})$-scaling scheme.
    
\end{proof}

\subsection{Analysis} \label{sec:1/6-analysis}

Note that Alice and Bob only ever append to $U_A, U_B, P_A, P_B$, and once a symbol has been appended it is never modified. Thus, throughout the analysis, when we refer to $U_A, U_B, P_A, P_B$, we mean their values at the end of the protocol, so that $U_A, U_B \in (\{0,1,\rewind,\bullet\}^2)^K$ and $P_A,P_B\in(\Sigma^2)^K$.

\subsubsection{Unique Decoding Lemma}

\begin{definition}[$\cS$]\label{def:S}
    We define the set $\cS$ to consist of all rounds $k\in [K]$ where one of the following conditions does \emph{not} hold.
    \begin{enumerate}[label={(\roman*)}]
        \item For not necessarily distinct parties $P,P'\in \{A,B\}$, it holds that $\PikC(U_P)[k]=P_{P'}[k]\in \Sigma^2 \implies \PikCDec(P_{P'}[1:k])= v(U_P[1:k])$.
        \item $\PikC(U_A)[k]=\PikC(U_B)[k] \in \Sigma^2 \implies v(U_A[1:k])=v(U_B[1:k]).$
    \end{enumerate}
\end{definition}

\begin{lemma}\label{lem:S}
    $\cS$ has size at most $20\epsilon K$.
\end{lemma}

\begin{proof}
We deal with each of the conditions individually.
\begin{enumerate}[label={(\roman*)}]
    \item Let $\cS_1$ be the set of indices that violate the first condition. For each pair of parties $P,P'$, by Theorem~\ref{thm:listGcode-decoding}, it holds that there are only $2\epsilon \cdot 2K$ values of $k$ where $\PikC(U_P)[k]=P_{P'}[k]\implies \PikC(U_P)[k][2]=P_{P'}[k][2]$,\footnote{Recall that $\PikC(U_P)[k], P_{P'}[k]\in \Sigma^2$ so $\PikC(U_P)[k][2], P_{P'}[k][2]\in \Sigma$.} but $\PikCDec(P_{P'}[1:k])\neq v(U_P[1:k]).$ Thus, adding over all four cases of $P,P'\in \{A,B\}$, it holds that $\cS_1$ has size at most $4\cdot 2 \epsilon 2K= 16\epsilon K$.
    
    \item Let $\cS_2$ be the set of indices that violate the second condition. By Theorem~\ref{thm:listGcode-decoding}, it holds that there are only $2\epsilon \cdot 2K$ values of $k$ where $\PikC(U_A)[k]=\PikC(U_B)[k]\implies \PikC(U_A)[k][2]=\PikC(U_B)[k][2]$ but $v(U_A[1:k])\neq \PikCDec(\PikC(U_B[1:k])$. The latter is always either $v(U_B[1:k])$ or $\perp$, so there are at most $2\epsilon \cdot 2K$ values of $k$ where $v(U_A[1:k])\neq v(U_B[1:k])$. Thus, $\cS_2$ is size at most $4\epsilon K$. 
\end{enumerate}
    
    The total size of $\cS$ is at most $|\cS_1|+|\cS_2|\leq 20\epsilon K$.
\end{proof}

\subsubsection{Definitions for the Potential}

To prove Theorem~\ref{thm:1/6}, we analyze the effects of corruption on the \emph{good} and \emph{bad updates} Alice/Bob make. We begin by defining good, bad, and neutral updates. After receiving a message from Bob, Alice updates her transcript $U_A$ and confidence $w_A$ to $U'_A$ and $w'_A$.
\begin{itemize}
    \item Let $(\mathcal{U}'_A, \mathcal{W}'_A)=(U_A,w_A)\otimes_x\op_{\cT}(t(v(U_A)))$. The update is good if $t(v(\mathcal{U}'_A))=t(v(U'))$ and $\mathcal{W}'_A=w_A$.
    \item The update is neutral if $(t(v(U'_A)), w'_A)=(t(v(U_A)),w_A)$.
    \item The update is bad otherwise.
\end{itemize}
We similarly define good and bad updates for Bob. We will often refer to making a good/bad update as simply \emph{making an update}, and considering a neutral update as having done nothing.

For each $t \in [1, \dots, K]$, we define the following potential functions:
\begin{itemize}
    \item $\psi^A_t$ is defined to be the total number of good updates minus the number of bad updates Alice has done in response to messages $1, \dots, t$. Note that she only updates in response to messages she receives (the even numbered messages). 
    \item $\psi^B_t$ is defined to be the total number of good updates minus the number of bad updates Bob has done in response to messages $1, \dots, t$. Note that he only updates in response to messages he receives (the odd numbered messages). 
\end{itemize}

\begin{lemma} \label{lemma:N}
    The potential $\psi_t^A$ determines Alice's final transcript guess and her confidence as follows:
    \begin{enumerate}[label={(\roman*)}]
        \item If $\psi_t^A\geq n_0/2$, then $t(v(U_A))=\cT$ and $w_A\geq\psi_t^A-n_0/2$.
        \item If $\psi_t^A\leq n_0/2$, then $t(v(U_A))\neq\cT$ and $w_A\leq n_0/2-\psi_t^A$.
    \end{enumerate}
    The same statements hold for Bob, replacing $A$ with $B$.
\end{lemma}

\begin{proof}
We prove this for Alice as the proof for Bob is identical. After sending message $1$, since $U_A = \bullet 1$, in order make $t(v(U_A)) = \cT$, Alice needs to perform $n_0/2$ good updates (the first $n_0/2 - 1$ updates consist of appending two bits, corresponding to Bob's and her next messages in $\pi_0$, followed by $1$ further good update consisting of simply appending Bob's next message). Every good update thereafter increases $w_A$ by $1$ without changing $t(v(U_A))$.

It remains to show that every good update undoes a bad update; that is, every bad update, when followed by a good update, results back in the original value of $(t(v(U_A)),w_A)$. If the bad update appends two instructions $\in \{ 0, 1, \bullet \}^2 \backslash \{ \bullet \bullet \}$ to $U_A$, then the new value of $t(v(U_A))$ must not be a prefix of $\cT$. Then the next good instruction, which is $\rewind$, undoes this. If the bad update deletes the last one or two bits of $t(v(U_A))$ incorrectly, then re-appending the bit(s) undoes this. If the bad update increases $w_A$ incorrectly, then $t(v(U_A)) \not= \cT$, so the next good update is $\op_{\cT}(t(v(U_A)))$ which causes $w_A$ to decrease by $1$. If the bad update decreases $w_A$ incorrectly, then $t(v(U_A)) = \cT$, and the next good update is $\op_{\cT}(t(v(U_A)))$ which increases $w_A$ by $1$.
\end{proof}

From this point on, we will focus on analyzing Protocol~\ref{prot:1/6} from Alice's perspective, as the analysis from Bob's perspective follows analogously.

Define $\rho^A_t$ as follows (and similarly $\rho^B_t$): $\rho^A_t$ is the \emph{expected} number of good updates minus the number of bad updates that Alice will do in response to message $t$, given the protocol so far, if message $t$ is uncorrupted. (Note that $\rho^A_t = 0$ for odd $t$ since Alice sends the odd messages.)

Define $\val^A_t$ as follows:
\[
\val^A_t=\begin{cases}
0.5 & \text{if $t$ is odd and message $t$ is of the form $\ECC(z\in\Sigma^2,?)$ and ($\psi^A_t < n_0/2$ or $\psi^B_{t-1}\geq n_0/2$).}\\
0.5 & \text{if $t$ is even and message $t$ is of the form $\ECC(z\in\Sigma^2,?)$ and $\psi^B_t<n_0/2$.}\\
0 & \text{otherwise}
\end{cases}
\]

Define the potential $\Psi^A_t$ as follows:
\[
    \Psi^A_t=\psi^A_t+\rho^A_{t+1}+\min(\psi^B_t+\rho^B_{t+1}, n_0/2) + \val^A_{t+1}
\]

Finally, we define Alice's actual update: $\Lambda^A_t$ is the actual value of the update Alice makes in response to message $t$ (in particular, $\Lambda^A_t \in \{ -1,0,1 \}$). 

Throughout the analysis, we say Alice \emph{interprets} a message $m$ as $\ECC(z^*,\delta^*)$ in Protocol~\ref{prot:1/6} when she enters Case~\ref{case1} or Case~\ref{case2} according to that value. Additionally, we will say she interprets the message correctly or incorrectly, if $\ECC(z^*,\delta^*)$ respectively equals or does not equal the message Bob sent.

\begin{lemma} \label{lem:basic-bounds}
The following are true for any $k\notin \cS$:
\begin{enumerate}
    \item $\rho^A_k\geq 0$. As a corollary, if Alice correctly interprets message $k$, then $\Lambda^A_k\geq 0$.
    \item For any $k$, it holds that $\bbE[\Lambda^A_k]-\rho^A_k\geq -3\alpha_k-3\epsilon$.
    \item For all even $k$, if Alice interprets message $k$ incorrectly, then $\bbE [\Lambda^A_k]\geq 0.5-3\alpha_k-3\epsilon$. Similarly, for all odd $k$, if Bob interprets message $k$ incorrectly, then $\bbE [\Lambda^B_k]\geq 0.5-3\alpha_k-3\epsilon$.
    \item Whenever Alice sends $\ECC(z\in \Sigma^2,?)$ as message $k$, it holds that $\val^A_k+\rho^B_k\geq 0.5$.
    \item Whenever Bob sends $\ECC(z\in \Sigma^2,?)$ as message $k$, it holds that $\val^A_k+\rho^A_k\geq 0.5$.
\end{enumerate}
\end{lemma}

\begin{proof}
    We prove the statements individually.
    \begin{enumerate}
        \item We assume Alice interprets Bob's message in the $k$'th round correctly. Let Bob's intended message be $\ECC(z,\delta)$. If $\delta=?$, then $z=\PikC(U_B)[k]$. We have $z=\PikC(U_B)[k]=P_A[k]$, so by Lemma~\ref{lem:S}, $v(U_B[1:k])=\PikCDec(P_A[1:k])$. Then, if Alice enters Case~\ref{case1}, $\PikCDec(P_A[1:k])=v(U_A[1:k])$ as well, so $v(U_A[1:k])=v(U_B[1:k])$. Since they are the same, they must be either $\emptyset$ or $\cT$. In either case, $\hat{\delta}=1$ results in a positive update. If Alice enters Case~\ref{case2}, then in order to have made an update, she must enter Case~\ref{case2} Subcase~\ref{sub22}, which she only enters if $v(U_B[1:k])$ is complete and consistent with her input, and therefore $=\cT$, resulting in a positive update. 
        
        If $\delta\in\{0,1,\rewind\}$, Bob sent $\ECC(P_B[k],\delta)$. The only way that Alice can make an update is by entering Case~\ref{case1}. This requires $P_B[k]=C(U_A[1:k])[k]\implies \PikCDec(P_B[1:k])=v(U_A[1:k])$.
        Note also that Bob must have decoded $\PikCDec(P_B[1:k-1])$ to $v^*$ and set $P_B[k] \gets \PikC(v^*, \bullet\bullet)$. Then, $\PikCDec(P_B[1:k]) \in \{ v^* \oplus \bullet \oplus \bullet, \perp \}$. Since $\PikCDec(P_B[1:k]) = v(U_A[1:k]) \not= \perp$, it holds that $\PikCDec(P_B[1:k]) = v^* \oplus \bullet \oplus \bullet \implies v(U_A[1:k]) = v^* \oplus \bullet \oplus \bullet)$. This means that Bob sends an instruction which causes Alice to make a positive update.
        
        To show $\Lambda^A_k\geq 0$, Alice either makes the update corresponding to the case she is in, or no update at all. In order for $\rho^A_k\geq 0$, this one possible update she could make must be a good update, so $\Lambda^A_k\geq 0$ as well.
        
        \item Clearly, if $k$ is odd, then $\Lambda^A_k - \rho^A_k = 0 \ge -3\alpha_k - 3\epsilon$. We focus on when $k$ is even. Let Bob's intended message be $\ECC(z\in\Sigma^2,\delta\in\{0,1,\rewind,\bullet\})$ 
        
        
        We split the proof into cases.
        
        \begin{caseof}
        \begin{mdframed}[topline=false, rightline=false, bottomline=false]
            \case{Alice does not enter Case~\ref{case1} or Case~\ref{case2} Subcase~\ref{sub22}.}{
                Alice does not update, so $\Lambda^A_k=0$. If Bob's message was of the form $\ECC(z_A,\delta)$, then $\rho^A_k \leq 1$ and $\alpha_k\geq \frac13$ (otherwise Alice should have entered Case~\ref{case1}). This gives
                \begin{align*}
                    &~ \bbE[\Lambda^A_k]-\rho^A_k \\
                    \geq&~ 0-1 \\
                    \geq&~ -3\alpha_k-3\epsilon.
                \end{align*}
                Otherwise if Bob's message was of the form $\ECC(z^*\neq z_A\in\Sigma^2,?)$, then $\alpha_k\geq \frac16-\epsilon$. He must enter Case~\ref{case2} or Case~\ref{case3}, so his expected update is at most $0.5$. Then, 
                \begin{align*}
                    &~\bbE[\Lambda^A_k]-\rho^A_k \\
                    \geq&~ 0-0.5 \\
                    \geq&~ -3\alpha_k-3\epsilon.
                \end{align*}
            }
            \case{Alice interprets message $k$ correctly and she enters Case~\ref{case1} or Case~\ref{case2}.}{
                We have $d_m\leq\alpha_k$. 
                We only need to look at the case where her possible update is positive; if it is $0$, the result follows from the calculation above and cannot be negative. If she enters Case~\ref{case1}, her probability of updating is $1-3d_m\geq 1-3\alpha_k$, so
                \begin{align*}
                    &~ \bbE[\Lambda^A_k]-\rho^A_k \\
                    \geq&~ (1-3\alpha_k) -1 \\
                    \geq&~ -3\alpha_k-3\epsilon.
                \end{align*}
                If she enters Case~\ref{case2} Subcase~\ref{sub22}, her probability of updating is $0.5-3d_m\geq 0.5-3\alpha_k$, so 
                \begin{align*}
                    &~ \bbE[\Lambda^A_k]-\rho^A_k \\
                    \geq&~ (0.5-3\alpha_k) -0.5 \\
                    \geq&~ -3\alpha_k-3\epsilon.
                \end{align*}
            }
            \case{Alice interprets message $k$ incorrectly as $\ECC(z^*,\delta^*)$ and enters Case~\ref{case1} or Case~\ref{case2} Subcase~\ref{sub22}.}{
                If she enters Case~\ref{case1} and $z=z^*$, then $d_m\geq \frac23-\alpha_k$ so her probability of updating is $1-3d_m\leq -1+3\alpha_k$, so
                \begin{align*}
                    &~ \bbE[\Lambda^A_k]-\rho^A_k \\
                    \geq&~ -1(-1+3\alpha_k) -1 \\
                    \geq&~ -3\alpha_k-3\epsilon.
                \end{align*}
                If she enters Case~\ref{case1} and $z\neq z^*$, then $d_m\geq \frac12-\epsilon-\alpha_k$, so her probability of updating is $1-3d_m\leq -0.5+3\alpha_k+3\epsilon$. Also, $\rho^A_k\leq 0.5$. This gives
                \begin{align*}
                    &~ \bbE[\Lambda^A_k]-\rho^A_k \\
                    \geq&~ -1(-0.5+3\alpha_k+3\epsilon) -0.5 \\
                    \geq&~ -3\alpha_k-3\epsilon.
                \end{align*}
                If she enters Case~\ref{case2} Subcase~\ref{sub22}, then $d_m\geq \frac12-\epsilon-\alpha_k$, so her probability of updating is $0.5-3d_m\leq -1+3\alpha_k+3\epsilon$.  This gives
                \begin{align*}
                    &~ \bbE[\Lambda^A_k]-\rho^A_k \\
                    \geq&~ -1(-1+3\alpha_k+3\epsilon) -1 \\
                    \geq&~ -3\alpha_k-3\epsilon.
                \end{align*}
            }
        \end{mdframed}
        \end{caseof}
        \item We prove this for Alice as the proof for Bob is symmetric. If $\rho^A_k\geq 0.5$, then the result follows from the previous item. Otherwise, $\rho^A_k=0$. Alice interprets message $k$ as $(z^*,\delta^*)$ and Bob's intended message was $(z,\delta)$, where $(z^*,\delta^*)\neq (z,\delta)$.
        
        If she enters Case~\ref{case1} and $z=z^*$, then $d_m\geq \frac23-\alpha_k$ so her probability of updating is $1-3d_m\leq -1+3\alpha_k$, so
        \begin{align*}
            &~ \bbE[\Lambda^A_k] \\
            \geq&~ -1(-1+3\alpha_k) \\
            \geq&~ 1-3\alpha_k-3\epsilon.
        \end{align*}
        If she enters Case~\ref{case1} and $z\neq z^*$, then $d_m\geq \frac12-\epsilon-\alpha_k$, so her probability of updating is $1-3d_m\leq -0.5+3\alpha_k+3\epsilon$. Also, $\rho^A_k\leq 0.5$. This gives
        \begin{align*}
            &~ \bbE[\Lambda^A_k] \\
            \geq&~ -1(-0.5+3\alpha_k+3\epsilon) \\
            \geq&~ 0.5-3\alpha_k-3\epsilon.
        \end{align*}
        If she enters Case~\ref{case2} Subcase~\ref{sub22}, then $d_m\geq \frac12-\epsilon-\alpha_k$, so her probability of updating is $0.5-3d_m\leq -1+3\alpha_k+3\epsilon$.  This gives
        \begin{align*}
            &~ \bbE[\Lambda^A_k] \\
            \geq&~ -1(-1+3\alpha_k+3\epsilon) \\
            \geq&~ 1-3\alpha_k-3\epsilon.
        \end{align*}
                
        \item Alice sends the odd messages, so we are in the case where $k$ is odd. If $\psi^A_t< n_0/2$ or $\psi^B_{t-1}\geq n_0/2$, then the result follows because $\val^A_k = 0.5$ and $\rho^B_k\geq 0$. Otherwise $\psi^A_k=\psi^A_{k-1}\geq n_0/2$. Thus, Alice's message is $\ECC(\PikC(U_A)[k],?)$ where $t(v(U_A)) = \cT$. If Bob receives this message uncorrupted, then $\PikC(U_A)[k]=P_B[k]$, so by Definition~\ref{def:S}, $v(U_A[1:k])=\PikCDec(P_B[1:k])$. If he enters Case~\ref{case1}, then $\PikC(U_A)[k] = \PikC(U_B)[k] \implies \cT = t(v(U_A[1:k])) = t(v(U_B[1:k]))$ so it must be the case that he makes a good update. If he enters Case~\ref{case2}, he decodes $v^*$ such that $t(v^*)=\cT$, and so also makes a good update with at least $0.5$ probability.
        \item The proof is very similar. Bob sends the odd messages, so we are in the case where $k$ is even. If $\psi^B_t< n_0/2$, then the result follows because $\psi^A_k = 0.5$ and $\rho^A_k\geq 0$. Otherwise $\psi^B_k=\psi^B_{k-1}\geq n_0/2$. Thus, Bob's message is $\ECC(\PikC(U_B)[k],?)$ where $t(v(U_B)) = \cT$. If Alice receives this message uncorrupted, then $\PikC(U_B)[k]=P_A[k]$, so by Definition~\ref{def:S}, $v(U_B[1:k])=\PikCDec(P_A[1:k])$. If she enters Case~\ref{case1}, she makes a good update, and if she enters Case~\ref{case2}, she decodes $v^*$ such that $t(v^*)=\cT$, and so also makes a good update with at least $0.5$ probability.
    \end{enumerate}
\end{proof}

\subsubsection{Calculating the Change in Potential}

The main objective is to prove the following lemma.

\begin{lemma}\label{lem:potential-change}
For any $k\in [K]$ such that $k-1,k,k+1\notin \cS$, if an $\alpha_{k}$ fraction of message $k$ is corrupted, then
\[
    \bbE [\Psi^A_k - \Psi^A_{k-1}] \ge 0.5 - 3\epsilon - 3\alpha_{k}.
\]
\end{lemma}

\begin{proof}


We split the proof into four parts depending on the parity of $k$ and on the value of $\psi^B_k$ or $\psi^B_{k-1}$.



\paragraph{$k$ is even and $\psi^B_k < n_0/2$.}

Then 
\begin{align*}
    &~ \bbE [\Psi^A_k - \Psi^A_{k-1}] \\
    =&~ \bbE[\psi^A_k+\rho^A_{k+1}+\min(\psi^B_k+\rho^B_{k+1}, n_0/2) + \val^A_{k+1}-\psi^A_{k-1}-\rho^A_{k}-\min(\psi^B_{k-1}+\rho^B_k, n_0/2) - \val^A_k] \\
    =&~ \bbE[\Lambda^A_k-\rho^A_{k} + \val^A_{k+1} - \val^A_k+\min(\psi^B_k+\rho^B_{k+1}, n_0/2)-\min(\psi^B_k, n_0/2)] \\
    =&~ \bbE[\Lambda^A_k-\rho^A_{k} + \val^A_{k+1} - \val^A_k+\rho^B_{k+1}]. 
\end{align*}
\begin{caseof}
    \begin{mdframed}[topline=false,rightline=false,bottomline=false]
    \case{Message $k$ is of the form $\ECC(z\in \Sigma^2,?)$.}{
        Notice that $z=\PikC(U_B)[k]$. 
        
        It holds that $\val^A_k= 0.5$. If the message is uncorrupted, Alice must enter Case~\ref{case2} Subcase~\ref{sub23} because $\PikCDec(P_A[1:k])=v(U_B[1:k])\neq v(U_A[1:k])$ by Definition~\ref{def:S}. Alice only enters Case~\ref{case1} when $\PikCDec(P_A[1:k])=v(U_A[1:k])$. Thus, $\rho^A_k=0$ because Alice makes a neutral update.
        Thus, we need to show
        \[
            \bbE[\Lambda^A_k + \val^A_{k+1} +\rho^B_{k+1}]\geq 1-3\alpha_k-3\epsilon.
        \]
        \begin{subcaseof}
            \subcase{Alice interprets message $k$ correctly.}{
             Then we must be in Case~\ref{case2} Subcase~\ref{sub23} as shown earlier. Also, $\Lambda^A_k=0$. With probability at least $1-6\alpha_k$, Alice sends a message of the form $\ECC(\PikC(U_B)[k+1],\delta)$ upon computing $\PikCDec(P_A[1:k])=v(U_B[1:k])$. This results in $\rho^B_{k+1}=1$. Otherwise (with probability at most $6\alpha_k$), she sends $\ECC(\PikC(U_A)[k+1],?)$; then by Lemma~\ref{lem:basic-bounds} $\val^A_{k+1}+\rho^A_{k+1}\geq 0.5$. Overall, this evaluates to 
            \begin{align*}
                &~ \bbE[\Lambda^A_k + \val^A_{k+1} +\rho^B_{k+1}] \\
                =&~ 0+(1-6\alpha_k)(1)+6\alpha_k(0.5) \\
                =&~ 1-6\alpha_k+3\alpha \\
                \geq&~ 1-3\alpha_k-3\epsilon.
            \end{align*}
            }
            \subcase{Alice enters Case~\ref{case3}.}{
                $\Lambda^A_k=0$ and $\val^A_{k+1}+\rho^B_{k+1}\geq 0.5$ by Lemma~\ref{lem:basic-bounds}. Also, $\alpha_k\geq \frac16-\epsilon$. This gives
            \begin{align*}
                &~\bbE[\Lambda^A_k + \val^A_{k+1} +\rho^B_{k+1}] \\
                =&~ 0+0.5 \\
                \geq&~ 1-3\alpha_k-3\epsilon.
            \end{align*}
            }
            \subcase{Alice interprets message $k$ incorrectly as $\ECC(z_A,\delta\in\{0,1,\rewind,\delta\})$.} {
                We have $\bbE[\Lambda^A_k]\geq 0.5-3\alpha_k-3\epsilon$ by Lemma~\ref{lem:basic-bounds} regardless of whether $z_A=z$. Alice sends a message of the form $\ECC(z\in\Sigma^2,?)$ so $\val^A_{k+1}+\rho^B_{k+1}\geq 0.5$ by Lemma~\ref{lem:basic-bounds}. This gives
            \begin{align*}
                &\bbE[\Lambda^A_k + \val^A_{k+1} +\rho^B_{k+1}] \\
                =&~ 0.5-3\alpha_k-3\epsilon+0.5 \\
                \geq&~ 1-3\alpha_k-3\epsilon.
            \end{align*}
            }
            \subcase{Alice interprets message $k$ incorrectly as $(z^*,\delta)$ where $z^*\neq z_A$.}{
                Let $d_m$ be the relative distance from the received message to $\ECC(z^*,\delta)$. Notice that Alice updates with probability $0.5-3d_m\leq 0.5-3(0.5-\epsilon-\alpha_k)=-1+3\alpha_k+3\epsilon$ probability, so
                \begin{align*}
                    &~\bbE[\Lambda^A_k + \val^A_{k+1} +\rho^B_{k+1}] \\
                    \geq&~ \Lambda^A_k \\
                    \geq&~ -1(-1 + 3\alpha + 3\epsilon) \\
                    \geq&~ 1-3\alpha_k-3\epsilon.
                \end{align*}
            }
        \end{subcaseof}
    }
    \case{Message $k$ is of the form $\ECC(z,\delta)$ for some $\delta\in \{0,1,\rewind\}$.}{
        We have $\val^A_k=0$ because $\delta\neq ?$. Thus, we need to show
        \[
            \bbE[\Lambda^A_k -\rho^A_k + \val^A_{k+1} +\rho^B_{k+1}]\geq 0.5-3\alpha_k-3\epsilon.
        \]
        \begin{subcaseof}
            \subcase{Alice enters any case except Case~\ref{case2} Subcase~\ref{sub23}.}{
                We have $\bbE[\Lambda^A_k]-\rho^A_k\geq -3\alpha_k-3\epsilon$ by Lemma~\ref{lem:basic-bounds} and $\val^A_{k+1}+\rho^B_{k+1}\geq 0.5$ by Lemma~\ref{lem:basic-bounds}. This gives
                \begin{align*}
                    &~ \bbE[\Lambda^A_k -\rho^A_k + \val^A_{k+1} +\rho^B_{k+1}] \\
                    \geq&~ -3\alpha_k-3\epsilon + 0.5 \\
                    =&~ 0.5-3\alpha_k-3\epsilon.
                \end{align*}
            }
            \subcase{Alice enters Case~\ref{case2} Subcase~\ref{sub23}.}{
                $\Lambda^A_k=0$ because Alice does not update. Also $\rho^A_k\leq 1$. Alice must have interpreted incorrectly since the received message has $\delta = ?$, so with probability of at least $1-p\geq 6(0.5-\epsilon-\alpha_k)$, Alice sends a message of the form $\ECC(z\in\Sigma^2,?)$, where $\val^A_{k+1}+\rho^B_{k+1}\geq 0.5$. This gives
                \begin{align*}
                    &~ \bbE[\Lambda^A_k -\rho^A_k + \val^A_{k+1} +\rho^B_{k+1}] \\
                    \geq&~ 0 - 1 +6(0.5-\epsilon-\alpha_k)\cdot 0.5 + =\\
                    \geq&~ 0.5-3\alpha_k-3\epsilon.
                \end{align*}
            }
        \end{subcaseof}
    }
    \end{mdframed}
\end{caseof}

\paragraph{$k$ is even and $\psi^B_k \geq n_0/2$.}

Then
\begin{align*}
    &~\bbE [\Psi^A_k - \Psi^A_{k-1}] \\
    =&~ \bbE[\psi^A_k+\rho^A_{k+1}+\min(\psi^B_k+\rho^B_{k+1}, n_0/2) + \val^A_{k+1}-\psi^A_{k-1}-\rho^A_{k}-\min(\psi^B_{k-1}+\rho^B_k, n_0/2) - \val^A_k] \\
    =&~ \bbE[\Lambda^A_k-\rho^A_{k} + \val^A_{k+1} - \val^A_k+\min(\psi^B_k+\rho^B_{k+1}, n_0/2)-\min(\psi^B_k, n_0/2)] \\
    =&~ \bbE[\Lambda^A_k-\rho^A_{k} + \val^A_{k+1} - \val^A_k].
\end{align*}

    We have that $\val^A_k=0$ because either the message is $\ECC(z\in\Sigma^2,?)$ with $\psi^B_k\geq n_0/2$, or $\ECC(z\in\Sigma^2,\delta\in\{0,1,\rewind\})$.
    Thus, we need to show
    \[
        \bbE[\Lambda^A_k-\rho^A_{k} + \val^A_{k+1}]\geq 0.5-3\alpha_k-3\epsilon.
    \]
\begin{caseof}
    \begin{mdframed}[topline=false,rightline=false,bottomline=false]
    \case{Alice does not enter Case~\ref{case2} Subcase~\ref{sub23}.}{
         We know $\Lambda^A_k-\rho^A_{k}\geq -3\alpha_k-3\epsilon$ by Lemma~\ref{lem:basic-bounds} and $\val^A_{k+1}=0.5$ because message $k+1$ is of the form $\ECC(z\in\Sigma^2,?)$. Then
         \begin{align*}
            &~ \bbE[\Lambda^A_k-\rho^A_{k} + \val^A_{k+1}] \\
            \geq&~ -3\alpha_k-3\epsilon+0.5\\
            \geq&~ 0.5-3\alpha_k-3\epsilon.
        \end{align*}
    }
    \case{Alice interprets message $k$ enters correctly and enters Case~\ref{case2} Subcase~\ref{sub23}.}{
        Bob must have sent $\ECC(\PikC(U_B)[k],?)$. It holds that $P_A[k]=\PikC(U_B)[k]$ so by Definition~\ref{def:S}, unless $k\in \cS$, $\PikCDec(P_A)=v(U_B)$. However, since $\psi^B_k\geq n_0/2$ she must have actually entered Case~\ref{case2} Subcase~\ref{sub22}, which is a contradiction.
    }
    \case{Alice interprets message $k$ incorrectly and enters Case~\ref{case2} Subcase~\ref{sub23}.}{
        $\Lambda^A_k=0$ because Alice does not update. Also, $\rho^A_k\leq 1$. With probability at least $1-p\geq 6(0.5-\alpha_k)$, Alice sends a message of the form $\ECC(z\in\Sigma^2,?)$, so $\val^A_{k+1}+\rho^B_{k+1}\geq 0.5$. This gives
        \begin{align*}
            & \bbE[\Lambda^A_k -\rho^A_k + \val^A_{k+1}] \\
            \geq&~ 0 - 1 +6(0.5-\alpha_k)\cdot 0.5\\
            \geq&~ 0.5-3\alpha_k-3\epsilon.
        \end{align*}
    }
    \end{mdframed}
\end{caseof}



\paragraph{$k$ is odd and $\psi^B_{k-1} < n_0/2$.}

Then the expression simplifies to
\begin{align*}
    &~ \bbE [\Psi^A_k - \Psi^A_{k-1}] \\
    =&~ \bbE[\psi^A_k+\rho^A_{k+1}+\min(\psi^B_k+\rho^B_{k+1}, n_0/2) + \val^A_{k+1}-\psi^A_{k-1}-\rho^A_{k}-\min(\psi^B_{k-1}+\rho^B_k, n_0/2) - \val^A_k] \\
    =&~ \bbE[\rho^A_{k+1} + \val^A_{k+1} - \val^A_k+\min(\psi^B_k, n_0/2)-\min(\psi^B_{k-1}+\rho^B_k, n_0/2)] \\
    =&~ \bbE[\rho^A_{k+1} + \val^A_{k+1} - \val^A_k+\psi^B_k-\psi^B_{k-1}-\rho^B_k] \\
    =&~ \bbE[\rho^A_{k+1}+\Lambda^B_k-\rho^B_k + \val^A_{k+1} - \val^A_k].
\end{align*}

\begin{caseof}
\begin{mdframed}[topline=false,rightline=false,bottomline=false]
\case{$\psi^A_k\geq n_0/2$ or message $k$ is of the form $\ECC(z\in\Sigma^2,\delta\in\{0,1,\rewind\})$.}{
We know that $\val^A_k=0$. Thus, we want to show
\begin{align*}
    \bbE[\Lambda^B_k-\rho^B_k +\rho^A_{k+1}+ \val^A_{k+1}] \geq 0.5-3\alpha_k-3\epsilon.
\end{align*}
    \begin{subcaseof} 
        \subcase{Bob does not enter Case~\ref{case2} Subcase~\ref{sub23}.}{
            Bob's next message is of the form $\ECC(z\in\Sigma^2,?)$ so $\rho^A_{k+1}+\val^A_{k+1}\geq 0.5$ by Lemma~\ref{lem:basic-bounds}. By the same lemma, $\bbE[\Lambda^B_k]-\rho^B_k\geq -3\alpha_k-3\epsilon$. This gives
            \begin{align*}
                &~ \bbE[\Lambda^B_k-\rho^B_k +\rho^A_{k+1}+ \val^A_{k+1}] \\
                \geq&~ 0.5-3\alpha_k-3\epsilon.
            \end{align*}
        }
        \subcase{Bob interprets message $k$ correctly and enters Case~\ref{case2} Subcase~\ref{sub23}.}{
            If message $k$ is of the form $\ECC(z\in\Sigma^2,\delta)$ for some $\delta\neq?$, Bob cannot have entered Case~\ref{case2}. Thus, $\psi^A_k\geq n_0/2$ and Alice must have sent $\ECC(\PikC(U_A)[k],?)$, and so $P_B[k]=\PikC(U_A)[k]$. Then by Definition~\ref{def:S}, $\PikCDec(P_B[1:k])=v(U_A[1:k])$, and since $\psi^A_k\geq n_0/2$, it holds that $t(\PikCDec(P_A[1:k]))=t(v(U_A))=\cT$. Then, Bob enters Case~\ref{case2} Subcase~\ref{sub22}, which is a contradiction.
        }
        \subcase{Bob interprets message $k$ incorrectly and enters Case~\ref{case2} Subcase~\ref{sub23}.}{
            $\Lambda^B_k=0$ and Bob sends $\ECC(z\in\Sigma^2,?)$ with probability $1-p\geq6(0.5-\epsilon - \alpha_k)$ resulting in $\val^A_{k+1}+\rho^A_{k+1}\geq 0.5$, so
            \begin{align*}
                &~ \bbE[\Lambda^B_k-\rho^B_k +\rho^A_{k+1}+ \val^A_{k+1}] \\
                \geq&~ 0-1+0.5(3-6\epsilon-6\alpha_k) \\
                =&~  0.5-3\alpha_k-3\epsilon.
            \end{align*}
        }
    \end{subcaseof}
}
\case{Message $k$ is of the form $\ECC(z\in\Sigma^2,?)$ and $\psi^A_k<n_0/2$.}{
    Note that $z=\PikC(U_B)[k]$ and we know that $\val^A_k=0.5$ and $\rho^B_k=0$. Thus, we need to show
    \begin{align*}
        \bbE[\Lambda^B_k +\rho^A_{k+1}+ \val^A_{k+1}] \geq 1-3\alpha_k-3\epsilon.
    \end{align*}
    \begin{subcaseof}
        \subcase{Bob interprets message $k$ correctly.}{
            Bob must enter Case~\ref{case2} Subcase~\ref{sub23}. This is because $v(U_B[1:k])\neq v(U_A[1:k])$, so Bob cannot enter Case~\ref{case1} by Definition~\ref{def:S}. Upon entering Case~\ref{case2}, he correctly decodes $\PikCDec(P_B[1:k])=v(U_A[1:k])$, causing him to enter Case~\ref{case2} Subcase~\ref{sub23}. Then, with $p\geq 1-6\alpha_k$, we have $\rho^A_{k+1}=1$, because Bob sends $\ECC(\PikC(U_A)[k+1],\delta)$, where $\delta$ is such that Alice would make a positive update upon entering Case~\ref{case1} if she interprets the message correctly. Otherwise $\rho^A_{k+1}+\val^A_{k+1}\geq 0.5$. By Lemma~\ref{lem:basic-bounds}, $\Lambda^B_k\geq0$, which gives
            \begin{align*}
                &~ \bbE[\Lambda^B_k +\rho^A_{k+1}+ \val^A_{k+1}] \\
                \geq&~ 1(1-6\alpha_k)+0.5(6\alpha_k)+0 \\
                \geq&~ 1-3\alpha_k-3\epsilon.
            \end{align*}
        }
        \subcase{Bob interprets message $k$ incorrectly and does not enter Case~\ref{case2} Subcase~\ref{sub23}.}{
            Notice $\Lambda^B_k>0.5-3\alpha_k-3\epsilon$ by Lemma~\ref{lem:basic-bounds}, and $\rho^A_{k+1}+\val^A_{k+1}\geq 0.5$ by Lemma~\ref{lem:basic-bounds} since he sends $\ECC(z\in\Sigma^2,?)$ in all cases except Case~\ref{case2} Subcase~\ref{sub23}. This gives
            \begin{align*}
                &~ \bbE[\Lambda^B_k +\rho^A_{k+1}+ \val^A_{k+1}] \\
                \geq&~ 0.5-3\alpha_k-3\epsilon + 0.5 \\
                \geq&~ 1-3\alpha_k-3\epsilon.
            \end{align*}
        }
        \subcase{Bob interprets message $k$ incorrectly and enters Case~\ref{case2} Subcase~\ref{sub23}.}{
            Notice $\Lambda^B_k=0$ and $\alpha_k\geq \frac13$.
            \begin{align*}
                &~ \bbE[\Lambda^B_k +\rho^A_{k+1}+ \val^A_{k+1}] \\
                \geq&~ 0+0+0 \\
                =&~  1-3\alpha_k-3\epsilon.
            \end{align*}
        }
    \end{subcaseof}
}
\end{mdframed}
\end{caseof}

\paragraph{$k$ is odd and $\psi^B_{k-1} \geq n_0/2$.}

Then
\begin{align*}
    &~ \bbE [\Psi^A_k - \Psi^A_{k-1}] \\
    =&~ \bbE[\psi^A_k+\rho^A_{k+1}+\min(\psi^B_k+\rho^B_{k+1}, n_0/2) + \val^A_{k+1}-\psi^A_{k-1}-\rho^A_{k}-\min(\psi^B_{k-1}+\rho^B_k, n_0/2) - \val^A_k] \\
    =&~ \bbE[\rho^A_{k+1} + \val^A_{k+1} - \val^A_k+\min(\psi^B_k, n_0/2)-\min(\psi^B_{k-1}+\rho^B_k, n_0/2)] \\
    \geq&~ \bbE[\rho^A_{k+1} + \val^A_{k+1} - \val^A_k+\min(\Lambda^B_k,0)].
\end{align*}

\begin{caseof}
    \begin{mdframed}[topline=false,rightline=false,bottomline=false]
    \case{Message $k$ is of the form $\ECC(z\in\Sigma^2,?)$.}{
        It holds that $z=\PikC(U_A)[k]$. Moreover, $\val^A_k=0.5$ since $\psi^B_{k-1}\geq n_0/2$, so we want to show
        \begin{align*}
            \bbE[\rho^A_{k+1} + \val^A_{k+1}+\min(\Lambda^B_k,0)] \geq 1-3\alpha_k-3\epsilon.
        \end{align*}
        \begin{subcaseof}
            \subcase{Bob interprets message $k$ correctly.}{
                If Bob entered Case~\ref{case1}, then $\PikC(U_A)[k]=\PikC(U_B)[k]$, which means $v(U_A[1:k])=v(U_B[1:k])$ by Definition~\ref{def:S}. If Bob entered Case~\ref{case2} Subcase~\ref{sub22}, then $v^*=v(U_A[1:k])=v(U_B[1:k])$ In either case, since $t(v(U_B[1:k]))=\cT$, Bob makes a neutral or positive update from his current complete correct transcript, so his next message is always $\ECC(\PikC(v(U_B[1:k]),\bullet\bullet),?)$ which has $\rho^A_{k+1}=1$. Also, $\Lambda^B_k\geq 0$ by Lemma~\ref{lem:basic-bounds}, so
                \begin{align*}
                    &~ \bbE[\rho^A_{k+1} + \val^A_{k+1}+\min(\Lambda^B_k,0)] \\
                    \geq&~ 1+0+0 \\
                    \geq&~  1-3\alpha_k-3\epsilon.
                \end{align*}
                
                If he entered Case~\ref{case2} Subcase~\ref{sub23}, he correctly decodes $v^*=v(U_A[1:k])$, and sends $\ECC(\PikC(U_A)[k],\delta\in\{0,1,\rewind,?\})$ with $\rho^A_{k+1}=1$ with probability at least $1-6\alpha_k$ and otherwise $\rho^A_{k+1}+\val^A_{k+1}\geq 0.5$. Also, $\Lambda^B_k\geq 0$ by Lemma~\ref{lem:basic-bounds}. This gives
                \begin{align*}
                    &~ \bbE[\rho^A_{k+1} + \val^A_{k+1}+\min(\Lambda^B_k,0)] \\
                    \geq&~ 1(1-6\alpha_k)+0.5(6\alpha_k)+0 \\
                    \geq&~  1-3\alpha_k-3\epsilon.
                \end{align*}
            }
            \subcase{Bob interprets message $k$ incorrectly.}{
                If Bob enters Case~\ref{case2} Subcase~\ref{sub23}, he never updates, in which case $\Lambda^B_k=0$. With probability at least $1-p\geq 6(0.5-\alpha_k-\epsilon)$, Bob sends $\ECC(z\in\Sigma^2,?)$, so $\rho^A_{k+1} + \val^A_{k+1}\geq 0.5$. This gives 
                \begin{align*}
                    &~ \bbE[\rho^A_{k+1} + \val^A_{k+1}+\min(\Lambda^B_k,0)] \\
                    \geq&~ 0.5\cdot 6(0.5-\alpha_k-\epsilon)+0 \\
                    =&~  1.5-3\alpha_k-3\epsilon.
                \end{align*}
                
                Otherwise, his probability of updating is at most $3\alpha_k+3\epsilon-0.5$, so $\bbE[\Lambda^B_k]\geq 0.5-3\alpha_k-3\epsilon$. Since he sends $\ECC(z\in\Sigma^2,?)$, we have $\rho^A_{k+1} + \val^A_{k+1}\geq 0.5$ which gives
                \begin{align*}
                    &~ \bbE[\rho^A_{k+1} + \val^A_{k+1}+\min(\Lambda^B_k,0)] \\
                    \geq&~ 0.5+0.5-3\alpha_k-3\epsilon \\
                    =&~  1-3\alpha_k-3\epsilon.
                \end{align*}
            }
        \end{subcaseof}
    }
    \case{Message $k$ is of the form $\ECC(z,\delta\in \{0,1,\rewind\})$.}{
        The message is not a question so $\val^A_k=0$. Thus, we need to show
        \begin{align*}
            \bbE[\rho^A_{k+1} + \val^A_{k+1} + \min(\Lambda^B_k,0)]\geq 0.5-3\alpha_k-\epsilon.
        \end{align*}
        
        \begin{subcaseof}
            \subcase{Bob interprets message $k$ correctly.} {He always sends a message $k+1$ of the form $\ECC(z,?)$, so $\rho^A_{k+1}+ \val^A_{k+1}\geq 0.5$. Then 
            \begin{align*}
                &~ \bbE[\rho^A_{k+1} + \val^A_{k+1} + \min(\Lambda^B_k,0)] \\
                \geq &~ 0.5-0 \\
                \geq &~ 0.5-3\alpha_k-\epsilon.
            \end{align*}
            }
            
            \subcase{Bob interprets message $k$ incorrectly.}{
            Notice that $\alpha_k\geq \frac16$ and so $\min(\Lambda^B_k,0)>0.5-3\alpha_k-3\epsilon$. Then
            \begin{align*}
                &~ \bbE[\rho^A_{k+1} + \val^A_{k+1} + \min(\Lambda^B_k,0)] \\
                \geq &~ 0-0.5-3\alpha_k-\epsilon \\
                = &~ 0.5-3\alpha_k-\epsilon.
            \end{align*}
            }
        \end{subcaseof}
    }
    \end{mdframed}
\end{caseof}

\end{proof}

\subsubsection{Concluding with Azuma's Inequality}

\begin{proof}[Proof of Theorem~\ref{thm:1/6}]
    We defer the proof of communication complexity and computational complexity to Lemma~\ref{lemma:cc}.
    Here, we simply show that Protocol~\ref{prot:1/6} is $\left(\frac16,1224\epsilon,2\cdot \exp \left(\frac{-\epsilon n_0}{800}\right)\right)$-scaling. First, the consistency property is clear: Alice never appends an operation to $U_A$ such that the resulting transcript $t(v(U_A))$ is inconsistent with $x$. It suffices to show the two scaling properties. In particular, we will show that with probability at least $1 - \exp \left( - \frac{\epsilon n_0}{800} \right)$, both of the following statements hold for Alice:
    \begin{itemize}
        \item If $\alpha<\frac16-1224\epsilon$, then $t(v(U_A))=\cT$ and $w_A\geq \frac{K}{2}(1-6\alpha-1224\epsilon)$.
        \item If $\alpha\geq \frac16-1224\epsilon$, then if $t(v(U_A))\neq \cT$ then $w_A\leq \frac{K}{2}(6\alpha-1+1224\epsilon)$.
    \end{itemize}
    We call these the Alice-scaling conditions. By a similar analysis, the equivalent statements will hold for Bob as well. Then a union bound will give that the probability the scaling conditions hold simultaneously for both parties is at least $1 - 2 \cdot \exp(- \frac{\epsilon n_0}{800})$.
    
    Let $\alpha_1, \dots, \alpha_{K}$ denote the fractional number of corruptions in messages $1, \dots, K$. Define $$\cS_k=\{i : i\leq k \wedge (i-1\in \cS \vee i\in \cS \vee i+1\in \cS)\}.$$ For $k \in \{1\dots K\}$, we define the random variables
    \begin{align*}
        \Phi^A_k &= \Psi^A_k - 0.5k + 3k\epsilon + \sum_{i=1}^{k} 3\alpha_i + 10|\cS_k|, \\
        \Phi^B_k &= \Psi^B_k - 0.5k + 3k\epsilon + \sum_{i=1}^{k} 3\alpha_i + 10|\cS_k|.
    \end{align*}
    
    By Lemma~\ref{lem:potential-change}, for all $k$ such that $k-1,k,k+1\notin \cS$, 
    \begin{align*}
        \bbE[\Phi^A_k] &= \bbE \left[ \Psi^A_k - 0.5k + 3k\epsilon + \sum_{i=1}^{k} 3\alpha_i + 10|\cS_k| \right] \\
        &\ge \bbE \left[ \Psi^A_{k-1} - 0.5(k-1) +3(k-1)\epsilon + \sum_{i=1}^{k-1} 3\alpha_i + 10|\cS_k| \right] \\
        &= \bbE[\Phi^A_{k-1}].
    \end{align*}
    For all $k$ such that either $k-1\in \cS$, $k\in \cS$, or $k+1 \in \cS$,
    \begin{align*}
        \bbE[\Phi^A_k] &~ = \bbE \left[ \Psi^A_k - 0.5k + 3k\epsilon + \sum_{i=1}^{k} 3\alpha_i + 10|\cS_k| \right] \\
        \ge&~ \bbE \left[ 
        \begin{aligned} 
            & \Psi^A_{k-1} + \Lambda^A_k + \rho^A_k - \rho^A_{k-1} +\min(\psi^B_k+\rho^B_{k+1}, n_0/2) -\min(\psi^B_{k-1}+\rho^B_{k}, n_0/2) \\
            & + \val^A_{k+1} - \val^A_{k} -0.5k +3k\epsilon + \sum_{i=1}^{k-1} 3\alpha_i + 10 |\cS_{k-1}| + 10 
        \end{aligned}
        \right]\\
        \ge&~ \bbE[\Phi^A_{k-1}] - \left| \Lambda^A_{k} \right| - \left| \Lambda^B_{k} \right| - \left| \rho^B_{k} \right| - \left| \rho^B_{k-1}
        \right| - \left| \rho^A_{k} \right| - \left| \rho^A_{k-1} \right| - \left| \val^A_{k+1} \right| - \left| \val^A_{k}
        \right| - 0.5 +3\epsilon + 3\alpha_k + 10 \\
        \ge&~ \bbE[\Phi^A_{k-1}].
    \end{align*}
    Therefore, $\{\Phi^A_k\}_{k\geq1}$ is a submartingale. A similar argument shows it has bounded distance 
    \begin{align*}
        |\Phi^A_k - \Phi^A_{k-1}| 
        &= \left| \Psi^A_k - \Psi^A_{k-1} - 0.5 +3\epsilon + 3\alpha_k + |\cS_k|-|\cS_{k-1}| \right| \\
        &\le \left| \Lambda^A_{k} \right| + \left| \Lambda^B_{k} \right| + \left| \rho^B_{k} \right| + \left| \rho^B_{k-1}
        \right| + \left| \rho^A_{k} \right| + \left| \rho^A_{k-1} \right| + \left| \val^A_{k+1} \right| + \left| \val^A_{k}
        \right| + \left|- 0.5 +3\epsilon + 3\alpha_{k} \right| + 10 \\
        &< 20.
    \end{align*}
    Similarly, $\Phi^B_k$ is a submartingale with bounded distance $< 20$. For convenience, define $\Phi^A_0=\Phi^B_0=-5$, and because $\Phi^A_1, \Phi^B_1 \in [-1, 15]$, it still holds that $\Phi^A$ and $\Phi^B$ are submartingales. Moreover, recall that $|\cS|\leq 20K\epsilon$ by Lemma~\ref{lem:S} which implies that $|\cS_K|\leq 60K\epsilon$.
    
    We now show that the Alice-scaling conditions hold as long as $\Psi^A_{K} \ge R:=n_0+2+\frac{K}{2}(1-6\alpha-1224\epsilon)$. Note that this implies that
    \begin{align*}
        \psi^A_K =&~ \Psi^A_K-\rho^A_{K+1}-\min(\psi^B_K+\rho^B_{K+1}, n_0/2) - \val^A_{K+1} \\
        \geq&~ \Psi^A_K-n_0/2-2 \\
        \geq&~ n_0/2 + \frac{K}{2} ( 1 - 6\alpha - 1224 \epsilon ).
    \end{align*}
    Then, by Lemma~\ref{lemma:N}, if $\alpha < \frac16 - 1224 \epsilon$, it holds that $\psi^A_K \ge n_0/2$ which means that Alice outputs $t(v(U_A)) = \cT$ with weight $w_A \ge \frac{K}{2} (1 - 6\alpha - 1224 \epsilon)$. On the other hand, if $\alpha \ge \frac16 - 1224 \epsilon$, then either $t(v(U_A)) = \cT$ or $\psi^A_K < n_0/2$, in which case $w_A \le n_0/2 - \psi^A_K \le \frac{K}{2}(6\alpha - 1 + 1224 \epsilon)$.
    
    Finally,
    \begin{align*}
        \Pr \left[ \Psi^A_{K} \ge R \right]
        &= 1 - \Pr \left[ \Phi^A_{K} - \Phi^A_0 < R - 0.5K + 3K\epsilon + \sum_{i=0}^{K} 3\alpha_i + 10|\cS_K| - \Phi^A_0 \right] \\
        &\ge 1 - \Pr \left[ \Phi^A_{K} - \Phi^A_0 < R - 0.5K + 3K\epsilon + 3 \alpha K +600K\epsilon + 5 \right] \\
        &\ge 1 - \Pr \left[ \Phi^A_K - \Phi^A_0 < n_0+2-\frac{K}{2}(1-6\alpha-1224\epsilon)-0.5K+3K\epsilon+3\alpha K+600K\epsilon+ 5 \right] \\
        &\ge 1 - \Pr \left[ \Phi^A_K - \Phi^A_0 < -n_0 \right] \\
        &\ge 1 - \exp \left( \frac{-\epsilon n_0}{800} \right).
    \end{align*}
    
    The same calculation holds for Bob. It follows that Protocol~\ref{prot:1/6} is $(\frac16, 1224 \epsilon, 2 \cdot \exp(-\frac{\epsilon n_0}{800}))$-scaling.
    
\end{proof}

\subsubsection{Communication and Computational Complexity}

\begin{lemma} \label{lemma:cc}
    The communication complexity of Protocol~\ref{prot:1/6} is $O_\epsilon(n_0)$, and the computational complexity is $2^{2^{O_\epsilon(n_0)}}$.
\end{lemma}

\begin{proof}
    The communication complexity is $K \cdot M(|\Sigma|, \epsilon) = O_\epsilon(n_0)$.
    
    As for the computational complexity, at the beginning, Alice and Bob agree on the code $\PikC$. Each possible code is defined by a labeling of $G$; there are $4 \cdot (2^K-1)$ edges with $|\Sigma|$ labels each, for $\le |\Sigma|^{4 \cdot 2^K}$ possible codes. Both Alice and Bob choose the lexicographically first one that is an $\epsilon$-sensitive layered code: $\epsilon$-sensitivity can be checked in time $\poly(|\Sigma|^K)$ by checking each word $w \in \Sigma^K$ and all possible prefix decodings. In each of the $K$ rounds, the substantial actions that Alice (respectively Bob) performs are some subset of the following:
    \begin{itemize}
        \item Alice appends elements in $\{0,1,\rewind,\bullet\}^2$ to $U_A$ or appends elements in $\Sigma^2$ to $P_A$. These steps take time $\Tilde{O}_\epsilon(1)$.
        \item Alice encodes $\PikC(U_A)$. This step takes time $\Tilde{O}_\epsilon(n_0)$.
        \item Alice decodes $\PikCDec(P_A)$. She may need to test all $4^{K}$ possible paths, which could take time $\Tilde{O}_\epsilon(n_0)\cdot 4^{K}$.
        \item Alice decodes a message $m$ to the nearest $\ECC(z\in\Sigma^2,\delta\in\{0,1,\rewind,?\}$ and computes the distance between $m$ and $\ECC(z \in \Sigma^2, \delta \in \{0,1,\rewind,?\})$. Since $|\Sigma|$ and therefore the length of $m$ is a constant independent of $n_0$, these steps take time $O_\epsilon(1)$.
    \end{itemize}
    
    In combination, the steps take total computational complexity $2^{2^{O_\epsilon(n_0)}}$ (where recall that $K = n_0/\epsilon$).
\end{proof}

\section{Acknowledgments}

Rachel Yun Zhang is supported by an Akamai Presidential Fellowship.

\bibliographystyle{alpha}
\bibliography{refs}

\end{document}